\documentclass[11pt]{article}

\usepackage{fullpage}

\usepackage{dsfont}
\usepackage{latexsym}
\usepackage{natbib}
\usepackage{amsmath}
\usepackage{amsthm}
\usepackage{amssymb}
\usepackage{amsfonts}
\usepackage{mathrsfs}
\usepackage{tikz}
\usepackage{aliascnt}
\usepackage{pstricks}
\usepackage{pst-all}
\usepackage{pstricks-add}
\usepackage{pst-plot}
\usepackage{graphicx}
\usepackage{subfig}
\usepackage{enumerate}

\usepackage[pdfpagelabels,pdfpagemode=None]{hyperref}
\usepackage{algorithmic}
\usepackage{algorithm}

\newcommand{\STOC}[1]{}
\newcommand{\NOTSTOC}[1]{#1}
\newcommand{\IFSTOCELSE}[2]{#2}

\newcommand{\argmax}{\operatorname{arg\,max}}

\newtheorem{theorem}{Theorem}%

\newaliascnt{lemma}{theorem}
\newtheorem{lemma}[lemma]{Lemma}%
\aliascntresetthe{lemma}

\newaliascnt{claim}{theorem}
\aliascntresetthe{claim}

\newaliascnt{corollary}{theorem}
\newtheorem{corollary}[corollary]{Corollary}%
\aliascntresetthe{corollary}

\newaliascnt{proposition}{theorem}
\newtheorem{proposition}[proposition]{Proposition}%
\aliascntresetthe{proposition}

\newaliascnt{remark}{theorem}

\aliascntresetthe{remark}

\newaliascnt{algo}{procedure}
\aliascntresetthe{algo}

\theoremstyle{definition}
\newtheorem{definition}{Definition}

\newtheorem{example}{Example}

\newcommand{\AutoAdjust}[3]{\mathchoice{ \left #1 #2  \right #3}{#1 #2 #3}{#1 #2 #3}{#1 #2 #3} }
\newcommand{\Xcomment}[1]{{}}

\newcommand{\inteval}[1]{\Big[#1\Big]}
\newcommand{\InParentheses}[1]{\AutoAdjust{(}{#1}{)}}
\newcommand{\InBrackets}[1]{\AutoAdjust{[}{#1}{]}}
\newcommand{\Ex}[2][]{\operatorname{\mathbf E}_{#1}\InBrackets{#2}}
\newcommand{\Prx}[2][]{\operatorname{\mathbf{Pr}}_{#1}\InBrackets{#2}}
\def\prob{\Prx}
\def\expect{\Ex}

\newcommand{\super}[1]{^{(#1)}}
\newcommand{\dd}{\mathrm{d}}  
\newcommand{\given}{\;\mid\;}

\newcommand{\expost}[1]{\tilde{#1}}

\newcommand{\dominated}{\preceq}
\newcommand{\Dist}{\dist}
\newcommand{\outcomespace}{W}
\newcommand{\payment}{p}
\newcommand{\distover}[1]{\Delta\InParentheses{#1}}
\DeclareMathOperator{\ALLOC}{Alloc}
\DeclareMathOperator{\PAYMENT}{Payment}
\DeclareMathOperator{\REV}{Rev}
\newcommand{\Rev}[1]{\REV[#1]}
\DeclareMathOperator{\UB}{UB}
\DeclareMathOperator{\MARGREV}{MR}
\newcommand{\MargRev}[1]{\MARGREV[#1]}
\DeclareMathOperator{\PSEUDOMARGREV}{PMR}
\newcommand{\PSEUDOMargRev}[1]{\PSEUDOMARGREV[#1]}
\DeclareMathOperator{\RULE}{Outcome}

\newcommand{\ttoq}{\operatorname{Quant}}

\newcommand{\mech}{{\cal M}}

\newcommand{\toutcomestepopt}[1][\exquant]{\toutcome^{#1}}

\newcommand{\allocstepopt}[1][\exquant]{\talloc^{#1}}

\newcommand{\qallocmr}{\qalloc^{\text{\it MR}}}
\newcommand{\callocmr}{\calloc^{\text{\it MR}}}
\newcommand{\callocsmr}{\callocs^{\text{\it MR}}}
\newcommand{\cumcallocmr}{\cumcalloc^{\text{\it MR}}}
\newcommand{\cumqallocmr}{\cumalloc^{\text{\it MR}}}
\newcommand{\toutcomemr}{\toutcome^{\text{\it MR}}}

\newcommand{\typespacesize}{m}
\newcommand{\distmr}{G^{\text{\it MR}}}

\newcommand{\pralloc}{\pi}

\newcommand{\optj}{{j^*}}

\newcommand{\qallocalt}{z}
\newcommand{\revscalar}{R^*}
\newcommand{\lotteryset}{L}

\newcommand{\maxvv}{\psi_{\max}}
\newcommand{\sos}{\nu_{(2)}}
\newcommand{\threshutil}{\mu}
\newcommand{\vv}{\psi}
\newcommand{\vvi}[1][i]{\psi_{#1}}
\DeclareMathOperator{\ORR}{ORR}
\newcommand{\eventi}{{\cal E}_i}

\newcommand{\ironed}{\bar}
\newcommand{\constrained}{\hat}
\newcommand{\optconstrained}{\composed{\optimized}{\constrained}}
\newcommand{\optimized}{\starred}
\newcommand{\differentiated}[1]{#1'}
\newcommand{\fortype}{\tilde}

\newcommand{\starred}[1]{#1^\star}
\newcommand{\noaccents}[1]{#1}
\newcommand{\composed}[3]{#1{#2{#3}}}

\newcommand{\forexquant}[1]{#1^{\exquant}}

\newcommand{\newagentvar}[3][\noaccents]{%
\expandafter\newcommand\expandafter{\csname #2\endcsname}{#1{#3}}%
\expandafter\newcommand\expandafter{\csname #2s\endcsname}{#1{\boldsymbol{#3}}}%
\expandafter\newcommand\expandafter{\csname #2smi\endcsname}[1][i]{#1{\boldsymbol{#3}}_{-##1}}%
\expandafter\newcommand\expandafter{\csname #2i\endcsname}[1][i]{#1{#3}_{##1}}%
\expandafter\newcommand\expandafter{\csname #2ith\endcsname}[1][i]{#1{#3}_{(##1)}}%
}

\newagentvar{alloc}{x}
\newagentvar[\tilde]{talloc}{\alloc}

\newagentvar{quant}{q}
\newagentvar[\constrained]{exquant}{\quant}
\newagentvar[\constrained]{critquant}{\quant}  
\newagentvar[\optconstrained]{monoq}{\quant}

\newagentvar{qprice}{\price}
\newagentvar{qrev}{R}
\newagentvar[\ironed]{iqrev}{\qrev}

\newagentvar{qalloc}{\alloc}
\newagentvar{cumalloc}{X}
\newagentvar[\constrained]{calloc}{\qalloc}
\newagentvar[\composed{\forexquant}{\constrained}]{callocstep}{\qalloc}
\newagentvar[\forexquant]{qallocstep}{\qalloc}
\newagentvar[\constrained]{cumcalloc}{\cumalloc}
\newagentvar[\ironed]{ialloc}{\qalloc}
\newagentvar[\ironed]{icumalloc}{\cumalloc}

\newagentvar{outcome}{w}
\newagentvar[\fortype]{toutcome}{\outcome}
\newagentvar[\composed{\forexquant}{\fortype}]{toutcomestep}{\outcome}

\newagentvar{rev}{R}
\newagentvar[\differentiated]{marg}{\rev}
\newagentvar{rawrev}{P}
\newagentvar[\differentiated]{rawmarg}{\rawrev}

\newagentvar[\tilde]{pseudorev}{\rev}
\newagentvar[\tilde]{pseudorawrev}{\rawrev}

\newagentvar{typespace}{T}
\newagentvar{typesubspace}{S}

\newagentvar{type}{t}
\newagentvar{val}{v}

\newagentvar{Val}{V}

\newagentvar{price}{p}

\newagentvar[\constrained]{critval}{\val}
\newagentvar[\constrained]{reserve}{\val} 
\newagentvar[\optconstrained]{monop}{v}

\newagentvar{util}{u}

\newagentvar{strat}{s}
\newagentvar{bid}{b}

\newagentvar{virt}{\phi}
\newagentvar{qvirt}{\phi}
\newagentvar[\ironed]{ivirt}{\virt}

\newagentvar{dist}{F}
\newagentvar{dens}{f}
\newagentvar{hazard}{h}

\newagentvar[\ironed]{iprice}{\price}  
\newagentvar[\ironed]{ival}{\val}  

\newagentvar{ints}{{\cal I}}

\newcommand{\reals}{{\mathbb R}}

\newcommand{\stepdist}[1]{G^{#1}}

\newcommand{\optpricing}{ex ante optimal pricing}
\newcommand{\optpricings}{ex ante optimal pricings}
\newcommand{\pseudopricing}{ex ante pseudo pricing}
\newcommand{\pseudopricings}{ex ante pseudo pricings}

\newcommand{\patht}{{r}}

\let\paragraphwithoutperiod\paragraph
\renewcommand{\paragraph}[1]{\paragraphwithoutperiod{#1.}}

\title{The Simple Economics of Approximately Optimal Auctions\thanks{This work was done in part while all authors were at Northwestern
University.  The second author was supported by NSF grants CCF-0643934
and AF-0910940 at Cornell University, the remaining by NSF CCF-0830773 at
Northwestern University.}}

\author{Saeed Alaei \\
Cornell University \\
Dept.\@ of Computer Science \\
{\tt saeed.a@gmail.com} 
\and Hu Fu \\
Microsoft Research \\
New England Lab \\
{\tt hufu@microsoft.com} \\
\and Nima Haghpanah \\
Northwestern University \\
EECS Department \\
{\tt nima.haghpanah@gmail.com}
\and Jason Hartline \\
Northwestern University \\
EECS Department \\
{\tt hartline@northwestern.edu}
}

\begin{document}

\begin{titlepage}
\maketitle
\begin{abstract}
The intuition that profit is optimized by maximizing marginal revenue
is a guiding principle in microeconomics.  In the classical auction
theory for agents with linear utility and single-dimensional
preferences, \citet{BR89} show that the optimal auction of \citet{M81}
is in fact optimizing marginal revenue.  In particular Myerson's
virtual values are exactly the derivative of an appropriate revenue
curve.

This paper considers mechanism design in environments where the agents
have multi-dimensional and non-linear preferences.  Understanding good
auctions for these environments is considered to be the main challenge
in Bayesian optimal mechanism design.  In these environments
maximizing marginal revenue may not be optimal, and furthermore, there
is sometimes no direct way to implement the marginal revenue
maximization.  Our contributions are three fold: we characterize the
settings for which marginal revenue maximization is optimal (by
identifying an important condition that we call {\em revenue
  linearity}), we give simple procedures for implementing marginal
revenue maximization in general, and we show that marginal revenue
maximization is approximately optimal.  Our approximation factor
smoothly degrades in a term that quantifies how far the environment is
from an ideal one (i.e., where marginal revenue maximization is
optimal).  Because the marginal revenue mechanism is optimal for
well-studied single-dimensional agents, our generalization immediately
approximately extends many results for single-dimensional agents to
more general preferences.

Finally, one of the biggest open questions in Bayesian algorithmic
mechanism design is in developing methodologies that are not
brute-force in the size of the agent type space (usually exponential
in the dimension for multi-dimensional agents).  Our methods identify
a subproblem that, e.g., for unit-demand agents with values drawn from
product distributions, enables approximation mechanisms that are
polynomial in the dimension.
\end{abstract}

\thispagestyle{empty}
\end{titlepage}

\newpage

\section{Introduction}
\label{s:intro}

{\em Marginal revenue} plays a fundamental role in microeconomic
theory.  For example, a monopolist providing a commodity to two
markets each with its own concave revenue (as a function of the supply
provided to that market) optimizes her profit by dividing her total
supply to equate the marginal revenues across the two markets.
Moreover this central economic principle also governs classical
auction theory.  \citet{M81} characterizes profit maximizing
single-item auction as formulaically optimizing the {\em virtual
  value} of the winner; \citet{BR89} reinterpret Myerson's virtual
value as the marginal revenue of a certain concave revenue curve.

Because it is simple and intuitive, the Myerson-Bulow-Roberts approach
provides the basis for most of Bayesian auction theory.  Unfortunately
though, this theory has been limited to settings where agents have linear
single-dimensional preferences, i.e., where an agent's utility is
given by her value for service less her payment.  Consequently,
Bayesian auction theory is often similarly limited.  With more general
forms of agent preferences; especially multi-dimensionality, e.g., for
multi-item auctions, or non-linearity, e.g., risk aversion or budgets;
auction theory is complex, less versatile, and often not well
understood.

Our main result is to show that hidden under the complexity of optimal
mechanism design problems for agents with multi-dimensional and
non-linear (henceforth:\@ general) preferences is marginal revenue
maximization.  The approach of marginal revenue maximization
decomposes a multi-agent mechanism design problem as a composition of
simple single-agent mechanism design problems, specifically, from the
construction of the appropriate notion of revenue curves.  This new
approach for general preferences uncovers a condition we refer to as
{\em revenue linearity} that is satisfied by all linear
single-dimensional preferences and governs the performance of the
marginal revenue mechanism more generally.  When the single-agent
problems are revenue linear, marginal revenue maximization is optimal
and the Myerson-Bulow-Roberts mechanism generalizes exactly.  When the
single-agent problems are approximately revenue linear, marginal
revenue maximization is approximately optimal (though the composition
of the single-agent mechanisms to implement marginal revenue
maximization requires new techniques).  Finally, because the marginal
revenue approach is structurally similar to the classical approach,
many results from classical auction theory approximately and
automatically extend to general preferences.

A central result in classical auction theory is derived from an
interpretation the Myerson-Bulow-Roberts mechanism (i.e., for
maximizing marginal revenue) in the special case of symmetric agents.
Our generalization admits a similar interpretation.  In the classical
setting there is a single item for sale and agents with i.i.d.\@
values for it; in our setting there is a single item for sale which
the seller can configure in one of several ways and agents have
i.i.d.\@ values for each configuration, e.g., a car that can be
painted red or blue (importantly, the seller sets the configuration
and the buyer cannot change it).\footnote{The red-or-blue car example
  is slightly unnatural as a forward auction (i.e., when the
  auctioneer is selling); however, the analogous reverse auction
  (i.e., the auctioneer is buying) is an important problem in
  procurement.  For instance the government may wish to hire a
  contractor to build a bridge.  Contractors can build different kinds
  of bridges.  From bids of the contractors over the different bridges
  the auctioneer selects a kind of bridge to procure, which contractor
  to procure it from, and how much is to be paid.  Our results for
  reverse auctions are analogous to those for forward auctions;
  interested readers can find the details in
  \autoref{app:procurement}.}
\begin{description}
\item[Selling a car.] Classical auction theory says that (a) the
  optimal way to sell an object (henceforth: a car) to a single agent
  with value drawn from a uniform distribution on $[0,1]$ is to post a
  take-it-or-leave-it price of 1/2, (b) the optimal way to sell a car
  to one of multiple agents with uniformly distributed values is to
  run a second-price auction with reserve price 1/2, and (c) more
  generally the optimal way to sell the car to multiple agents with
  i.i.d.\@ values is to run the second price auction with the same
  reserve price that would be offered as a take-it-or-leave-it price
  to one agent (assuming the distribution satisfies a mild
  assumption).

\item[Selling a red-or-blue car.]  Consider selling a car that, on
  sale, can be painted one of two colors, red or blue.\footnote{In
    this example we give the reserve price for $m=2$ colors; however,
    with the appropriate reserve price these results hold for any
    number of colors.}  Our theory says that (a) the optimal way to
  sell a red-or-blue car to a single agent with values for the
  different colors each drawn independently and uniformly from $[0,1]$
  is to post a take-it-or-leave-it price of $\sqrt{1/3}$ for either
  color, (b) the optimal way to sell a red-or-blue car to one of
  multiple agents each with i.i.d.\@ uniform values for each color is
  to run the second-price auction with reserve $\sqrt{1/3}$ and allow
  the winning agent to choose her favorite color on sale, and (c) more
  generally to sell a red-or-blue car to one of multiple agents each
  with values drawn i.i.d.\@ (from a distribution that satisfies the
  same mild assumption as above) for each color, the second price
  auction with the reserve price equal to the same price that would be
  offered to a single agent is (at worst) a 4-approximation to the
  optimal auction.
\end{description}
It should be noted that reducing a multi-dimensional preference to a
single-dimensional preference by always selling the winning agent her
favorite color is very natural and practical; however, it is not
generally optimal.  For example, when the agent's values for each
color is distributed uniformly on $[5,6]$, the analysis of
\citet{Tha04} shows that the optimal auction does not sell the agent
her favorite color subject to a reserve (in fact, it is not even
deterministic).  However, many relevant distributions, including the
uniform distribution on $[5,6]$, satisfy the mild assumption
sufficient for the four approximation, above.

\paragraph{Approach}

We focus on {\em service constrained environments} where, in any
outcome the mechanism produces, each agent is either considered served
or unserved.  The designer has a feasibility constraint that governs
which subset of agents can be simultaneously served, but the other aspects
of the outcome, e.g., payments, are unconstrained.  This model allows
additional unconstrained attributes of the service (e.g., the color of
the car in the previous red-or-blue car example, or the grade or quality of a service).  We assume that the
space of mechanisms is closed under convex combination, which allows
for randomized mechanisms.

The agents in the mechanism have independently but not necessarily
identically distributed preferences (a.k.a., types).  We do not place
any assumption on the agent preferences other than they are expected
utility maximizers.  This includes the most challenging preference
models in Bayesian mechanism design such as multi-dimensionality, public
or private budgets, and risk-aversion (e.g., as given by a concave
utility function).

{\em Revenue curves} result from the following single-agent mechanism
design problem.  Consider a single agent with private type drawn from
a known distribution.  Via the taxation principle \citep[see
  e.g.,][]{W97} the outcomes of a mechanism, for all possible reports
the agent might make in the mechanism, can be viewed as a menu where
the agent selects her favorite outcome by making the appropriate
report.  This menu may contain outcomes that are randomized and for
this reason we refer to it as a {\em lottery pricing}.  Ex ante, i.e.,
in expectation over the distribution of the agent's type, a
lottery pricing induces a probability with which the agent receives an
outcome that corresponds to service, and an expected payment, i.e.,
revenue.

As every lottery pricing induces an ex ante service probability and
expected revenue, we can ask the optimization question of identifying
the lottery pricing with a given ex ante service probability that has
the highest expected revenue.  As a function of the ex ante service
probability this optimal revenue induces a {\em revenue curve}.
Important in the construction of revenue curves are the lottery
pricings, i.e., single-agent mechanisms, that give the optimal revenue
for each ex ante service probability.  As the space of (mechanisms and
hence) lottery pricings is closed under convex combination, the
revenue curves are always concave.  The marginal revenue curve is
given by the derivative of the revenue curve with respect to ex ante
service probability.

As discussed in the opening paragraph, the standard economic intuition
suggests that a monopolist splitting the sale of a commodity between
two markets should do so to equate marginal revenue.  There is an
intuitive algorithmic reinterpretation of this fact.  If we break the
allocation to each market into tiny pieces ordered by willingness to
pay and attribute to each piece the change in revenue from adding that
piece (i.e., the marginal revenue), then the total revenue of an
allocation is the sum of the marginal revenues of each piece.  A
simple algorithm for optimizing this \emph{surplus of marginal
  revenue} is to repeatedly allocate a tiny amount to the market that
has the highest marginal revenue at its current allocation (until the
good is totally allocated or marginal revenues are non-positive).
Clearly this results in a final allocation where the markets marginal
revenues are roughly equal as in the microeconomic interpretation.
This allocation is optimal.

The main contribution of this paper is a methodology for constructing
multi-agent mechanisms from the simple single-agent lottery pricings
that define the revenue curve.  The main task of such a construction
is to specify a method for combining the single-agent mechanisms into
a multi-agent mechanism that is both feasible with respect to the
service constraint and obtains good revenue.  

\begin{definition}
\label{d:marginal-revenue-family}
The family of {\em marginal revenue mechanisms} take the following
form:
\begin{enumerate}
\item \label{step:quantile} Map each agent's private type (which may
  lie in an arbitrary type space) to a {\em quantile} in $[0,1]$.
\item \label{step:marg-rev} Calculate the marginal revenue of each
  agent as the derivative of the revenue curve at her quantile.
\item \label{step:feasible} Select for service the set of agents that
  maximize the {\em surplus of marginal revenue}, i.e., the total
  marginal revenue of agents served, subject to feasibility.
\item \label{step:outcome} Calculate for each agent the appropriate
  non-service aspects of the outcome, e.g., payments.
\end{enumerate}
\end{definition}

Thus far in the discussion only Steps~\ref{step:marg-rev} and
\ref{step:feasible} should be clear.  The remaining steps are
non-trivial in general and a main issue that we will be resolving.
For the special case of the selling-a-car example; where the agents'
values are independently, identically, and uniformly drawn from the
$[0,1]$ interval; the marginal revenue mechanism is instantiated as
follows.

For an agent with value drawn uniformly from the $[0,1]$ interval, the
optimal lottery pricing for ex ante service probability $\exquant$ is
to post a take-it-or-leave-it price of $\critval = 1-\exquant$.  The
revenue from such pricing is the price times the probability that it
is accepted.  Therefore, the revenue curve is $\rev(\exquant) =
(1-\exquant) \times \exquant$ and the marginal revenue curve is its
derivative $\marg(\exquant) = 1 - 2\exquant$.

The optimal lottery for ex ante probability $\exquant$ serves the
agent if her value $\val$ is on interval $[1-\exquant,1]$.  This is the
strongest $\exquant$ measure of the values from the distribution.
This motivates, in Step~\ref{step:quantile}, mapping value $\val$ to
quantile $\quant = 1-\val$.  Composing this mapping from value to
quantile with the above mapping from quantile to marginal revenue
gives a mapping from value $\val$ to marginal revenue as $2\val - 1$.

For a single-item auction, in Step~\ref{step:feasible} the surplus of
marginal revenue is maximized by serving nobody if all have negative
marginal revenues and, otherwise, by serving the agent with the
highest marginal revenue.  As the agents are symmetric and marginal
revenue is monotone in value, equivalently, the highest-valued agent
wins as long as her value is at least $1/2$ (the value for
which marginal revenue is zero, i.e., solving $2\val-1=0$).

The appropriate calculation of payments for Step~\ref{step:outcome} is
the following.  All losers have payments equal to zero.  The payment
of a winner is the minimum value she could declare and still win in
Step~\ref{step:feasible}, i.e., it is the maximum of the the second
highest agent value and 1/2.

This auction, as claimed in the earlier discussion of the
selling-a-car example, is the second-price auction with reserve $1/2$.
Moreover, the mapping from value to marginal revenue is identical to
the virtual values in the derivation of \citet{M81}.

\paragraph{Results}
This paper generalizes the marginal-revenue approach for agents with
single-dimensional linear preferences \citep{BR89} to general
preferences.  Our main algorithmic contribution is to generalize
Steps~\ref{step:quantile} and~\ref{step:outcome} thereby allowing the
construction of service constrained multi-agent mechanisms from
single-agent ex ante lottery pricings.  There are a number of
challenges in this endeavor.  First, revenue equivalance does not hold
for general preferences (which is used in the proof of optimality for
single-dimensional preferences).\footnote{For single-dimensional
  linear preference agents, the revenue equivalence theorem states
  that any two auctions with the same allocation in expectation have
  the same expected revenue.} Second, there is not a natural ordering
on types for general preferences (making it difficult to map types to
quantiles).  Third, the set of agents served by the marginal revenue
mechanism may be randomized.  None of these issues are present for
single-dimensional linear preferences.

%

Orthogonal to the question of implementing the marginal revenue
mechanism for general preferences are questions of quantifying its
performance.  Via the Myerson-Bulow-Roberts analysis it is known that
for single-dimensional linear preferences, the marginal revenue
mechanism is optimal.  As a first step in understanding the
performance of the mechanism more generally we give a new derivation
of the optimality for single-dimensinal agents.  Our derivation
exposes a previously unobserved property of single-dimensional linear
preferences which we refer to as {\em revenue linearity}.  Generally,
i.e., beyond single-dimensional linear preferences, the optimality of
the marginal revenue mechanism is implied by revenue linearity.
Moreover, if the single-agent problems are $\alpha$-approximately
revenue linear,
then the marginal revenue
mechanism is an $\alpha$ approximation to the optimal mechanism.

Revisiting our red-or-blue car example above, (a) is a description of
the optimal unconstrained lottery pricing, (b) is a consequence of the
revenue-linearity of unit-demand preferences that are uniformly
distributed on a multi-dimensional hypercube, and (c) is a consequence
of 4-approximate revenue linearity for agents with unit-demand
preferences drawn from any product distribution.

One of the main benefits of considering the marginal revenue mechanism
for approximately optimal mechanism design is that, as its structure
is similar to optimal mechanisms for single-dimensional environments,
many results from the extensive single-dimensional mechanism design
literature can be easily generalized.  The following are some of the
most important consequences.
\begin{description}
\item[Algorithmic mechanism design.] When weighted optimization is
  hard we can replace an exact algorithm for weighted maximization
  (Step~\ref{step:feasible} of \autoref{d:marginal-revenue-family})
  with any approximation algorithm using either of the
  single-dimensional black-box reductions of \citet{HL10} and
  \citet{HKM11}.
\item[Sequential posted pricing.] Sequential posted pricing mechanisms
  of \citet{CHMS10} and \citet{Y11} that are approximately optimal for
  single-dimensional agents are approximately optimal for general
  agents (in the same service constrained environment) and the same
  approximation factor is guaranteed.  Moreover, these sequential
  posted pricing bounds give another bound on the approximation factor
  of the marginal revenue mechanism.  The marginal revenue mechanism
  is in fact optimal within a class of mechanisms that contains the
  sequential posted pricing mechanisms; therefore, its approximation
  factor is no worse.  As an example, for the single-item service
  constraint, a sequential posted pricing bound implies an
  $e/(e-1)$-approximation regardless of approximate revenue linearity
  of the single-agent problems.
\item[Simple versus optimal.] While our marginal revenue mechanism is
  already generally much simpler than the optimal mechanism, we can
  get even simpler approximation mechanisms by applying methods
  developed for single-dimensional preferences to prove that simple
  mechanisms approximate the marginal revenue mechanism.  In
  particular, in single-dimensional environments maximizing marginal
  revenue is more complex than simple reserve-price-based mechanisms,
  i.e., mechanisms that maximize welfare subject to a reserve price.
  Nonetheless, \citet{HR09} show that reserve-price-based mechanisms
  are often approximately optimal.  When uniform pricing is
  approximately optimal, e.g., in generalizations of the red-or-blue
  car example, these mechanisms extend to general preferences.
\item[Single-sample mechanisms.] Approaches above have been for
  Bayesian optimal mechanism design where the designer optimizes a
  mechanism given a distribution of preferences.  \citet{DRY10} relax
  the assumption that the distribution is known and show that a
  mechanism based on drawing a single sample from the distribution
  gives a good approximation to the Bayesian optimal mechanism.
  Again, the single-sample framework extends to general preferences
  for which uniform pricing is approximately optimal.
\end{description}

It is important to contrast the simplicity of the marginal revenue
approach with recent algorithmic results in Bayesian mechanism design
for general agent preferences.  Recently, \citet{AFHHM12} and
\citet{CDW12,CDW12b,CDW13} gave polynomial-time mechanisms for large
important classes of Bayesian mechanism deign problems; the former
considers general preferences in service constrained settings (as does
this paper) and the latter considers multi-dimensional additive
preferences.  The two main conclusions of these works is that (a)
optimal mechanisms continue to have weighted maximization at their
core, and (b) the appropriate weights (i.e., virtual values) are
stochastic and can be solved for as a convex optimization problem,
e.g., via the ellipsoid method, that takes into account the
feasibility constraint and the distribution over types of all agents.
(This latter result is simply because the space of mechanisms is
convex, any point in the interior of a convex set can be implemented
by a convex combination of vertices, and vertices correspond to
linear, a.k.a., weighted, optimization.)  There are a number of
important distinctions between our work and these algorithmic results.
First, the weights in our derivation have a natural economic
interpretation as marginal revenues.  Second, the weights in our
derivation can be found easily from solutions to the single-agent
lottery pricing problems and are not derived from the solution to an
additional multi-agent optimization problem.  Third, in most cases,
the weights in our derivation depend only on the single-agent problem
and not on the multi-agent feasibility constraint or presence of other
agents.  Therefore, our approach affords significant structural
simplification and interpretation that enables the consequences
previously enumerated.  Finally, one of the biggest open questions in
the above algorithmic work is in developing approaches that are not
brute-force in each agent's type space.  As an example that breaks
this barrier, our approach gives approximately optimal mechanisms for
multi-dimensional unit-demand agents with values from product
distributions; these mechanisms are easy to compute with a
computational complexity that scales linearly with the dimensionality
of the type space (i.e., logarithmicly in the size of the type space).

\paragraph{Organization}
In \autoref{sec:single-dimensional} we review the
Myerson-Bulow-Roberts single-dimensional linear agent model and give a
new proof that the marginal revenue mechanism is revenue optimal.  The
proof follows from an argument that for single-dimensional linear
agents a class of single-agent lottery pricing problems satisfies a
natural revenue-linearity property.  
In \autoref{sec:general} we formalize the service constrained model
for general preferences and generalize the marginal revenue derivation
to general preferences that satisfy the previously identified
revenue-linearity property.
In \autoref{sec:implementation} we give general methods for implementing the
marginal revenue mechanism (e.g.\@ Steps~\ref{step:quantile}
and~\ref{step:outcome}) for general preferences regardless of revenue
linearity,
and in \autoref{sec:approx} we show that approximate revenue
linearity, properly defined, implies approximate optimality.
In \autoref{sec:simple} we suggest numerous extensions
of results in the single-dimensional mechanism design literature to
general preferences that are direct consequences of the marginal
revenue mechanism framework.

\section{Warmup: Single-dimensional Linear Preference}
\label{sec:single-dimensional}

In this section we warm up by giving a new proof that the marginal
revenue mechanism is revenue optimal for agents with
single-dimensional linear preferences.  In this proof we will
introduce many concepts that make our generalization possible (which
were not present in previous proofs).  The basic approach is as
follows.  We formulate an important class of lottery pricing problems,
the solution to which define a revenue curve.  We show that
single-dimensional linear agents are {\em revenue linear} in the sense
that it is optimal to decompose the allocation to any agent as a
convex combination of the solutions to these lottery pricing problems.
Finally, we observe that this decomposition implies that the optimal
revenue can be expressed in terms of the \emph{surplus of marginal
  revenue}\/: the sum of derivatives of the revenue curves of agents
served evaluated at points corresponding to the agents' types.  The
marginal revenue mechanism optimizes this latter term pointwise and,
therefore, also in expectation.  In the interest of brevity we will
keep the discussion informal; many of the proofs in this section are
subsumed by generalizations in \autoref{sec:general} which are given
formally.

\paragraph{Model}
A single-dimensional linear agent has a private type (a.k.a.\@
valuation) $\val \in \reals_+$ drawn at random from a probability
distribution with cumulative distribution function $\dist$ and density
function $\dens$.  Let $(\qalloc,\price)$ denote the outcome of
receiving a good or service with probability $\qalloc$ and making
expected payment $\price$.  For such an outcome, an agent with type
$\val$ has a linear utility $\util = \val \qalloc - \price$.

The geometry of single-dimensional auction theory is more readily
apparent when we index an agent's private type by its strength
relative to the distribution.  Let $\Val(\quant) =
\dist^{-1}(1-\quant)$ be the {\em inverse demand curve}, i.e.,
$\Val(\exquant)$ is the posted price that would be accepted by the
$\exquant$ measure of highest-valued agents (and rejected by all
others).  The {\em quantile} of an agent is the measure of
higher-valued types, i.e., an agent with type $\val$ has quantile
$\quant = 1-\dist(\val) = \Val^{-1}(\val)$.  Importantly, for $\val$
drawn at random from the distribution~$\dist$, $\quant =
\Val^{-1}(\val)$ is uniform on $[0,1]$ (therefore, expectations of
functions of $\quant$ are given by integrals with probability density
one).

A multi-agent mechanism design problem is given by $n$ such
single-dimensional agents, each with her respective inverse demand
curve (which may be distinct), and a feasibility constraint governing
the subsets of agents that can be simultaneously served.  E.g., for a
single-item auction, the feasibility constraint says that at most one
agent can be served; more generally, the feasibility constraint could
be given by a set system.  In the \emph{interim} stage, i.e., when an agent
knows her own value but not the values of other agents, the mechanism
looks to the agent like a single-agent mechanism.  It will thus be
sufficient for most of the analysis of optimal multi-agent mechanisms
to consider the appropriate single-agent problems.

From the perspective of an agent in a single-agent mechanism and as a
function of the agent's report, the agent is served with some
probability and makes some expected payment.  We can view this
function as a menu of service probabilities and expected payments
where the agent selects her favorite outcome by submitting the
corresponding report.  Notice that, depending on the agent's type, she
may choose different outcomes.  We may as well index the outcomes in
the menu by the quantile corresponding to the type for which the agent
would select the outcome, i.e., the agent with quantile $\quant$
chooses outcome $(\qalloc(\quant),\qprice(\quant))$.  We assume that
outcome $(\qalloc,\price) = (0,0)$ is in the menu.  This relabeling
and assumption imply {\em incentive compatibility} and {\em individual
  rationality}, respectively, i.e.,
\begin{align}
\label{eq:sd-IC}
\Val(\quant) \qalloc(\quant) - \qprice(\quant)
     &\geq \Val(\quant) \qalloc(\quant') - \qprice(\quant'),
     &\forall \quant,\quant' \in [0,1].
\tag{IC} \\
\label{eq:sd-IR}
\Val(\quant) \qalloc(\quant) - \price(\quant)
     &\geq 0,
     &\forall \quant \in [0,1].  \tag{IR}
\end{align}
We call such a menu a {\em lottery pricing}.  When the lottery pricing
is induced in the interim stage of a multi-agent mechanism, the
constraints above are {\em Bayesian incentive compatiblily} (BIC) and {\em
  interim individual rationality} (IIR).

The \citet{M81} characterization of Bayesian incentive compatible
mechanisms applies to lottery pricings and implies that the {\em
  allocation rule} $\qalloc(\cdot)$ is monotone non-increasing and the
{\em payment rule} $\price(\cdot)$ is given precisely as a function of
$\qalloc(\cdot)$.\footnote{Notice that quantiles are ordered in the
  opposite direction as types. Higher-valued types have low quantile
  and lower-valued types have high quantile.  Thus, the allocation
  rule should be non-increasing in quantile.}  An important
consequence of the latter part of this characterization is {\em
  revenue equivalence}.  We will make strong usage of both
monotonicity and revenue equivalence below, though the specific form
of the payment rule will not be important.

\paragraph{Constrained Lottery Pricings}
Given a lottery pricing and a distribution over the agent's value,
an ex ante expected payment $\expect[\quant]{\qprice(\quant)}$ and ex
ante probability of service $\expect[\quant]{\qalloc(\quant)}$ are
induced.  The single-agent lottery pricing problem that forms the
basis for the marginal revenue mechanism is the following.  Given an
ex ante constraint $\exquant$,
find the lottery pricing that serves the agent with
probability $\exquant$ and maximizes revenue. 

\begin{definition}
\label{def:revenue-curve}
The {\em revenue curve} $\rev(\exquant)$ is defined for all $\exquant
\in [0,1]$ as the revenue of the {\em ex ante optimal lottery pricing}
with allocation probability $\exquant$.
\end{definition}

To show that optimal mechanisms are convex combinations of ex ante
optimal lottery pricings, we consider a more general lottery pricing
problem.  Notice that the ex ante lottery pricing problem gives an
(equality) constraint on the total probability that the agent is
served in expectation over all quantiles she may have.  To get more
fine-grained control over the lottery pricing we additionally allow
upper bounds to be specified on the total probability of allocation to
subsets of quantiles.  Consider the following lottery pricing problem:
Given a monotone concave function $\cumcalloc(\quant)$, find the
optimal lottery pricing where the ex ante probability of allocating to
any $\exquant$ measure of quantiles is at most $\cumcalloc(\exquant)$
for all $\exquant \in [0,1)$ and exactly equal to $\cumcalloc(1)$ at
  $\exquant=1$.

To see why this constrained lottery pricing problem is the right one
to consider, notice the following.  First, because any allocation rule
is monotone, meaning stronger quantiles receive no lower probability
of service than weaker quantiles, the only set of measure $\exquant$
for which the constraint $\cumcalloc(\exquant)$ on service probability
may be tight is the strongest $\exquant$ measure of quantiles, i.e.,
$[0,\exquant]$.  For allocation rule $\qalloc(\cdot)$ the probability
of service to the strongest $\exquant$ measure of agents is exactly
$\cumalloc(\exquant) = \int_0^{\exquant} \qalloc(\quant)\,\dd\quant$.
We refer to $\cumalloc(\cdot)$ as the {\em cumulative allocation
  rule}.  Thus, the allocation constraint is exactly,
$\cumalloc(\exquant) \leq \cumcalloc(\exquant)$ for all $\exquant \in
[0,1]$ (with equality for $\exquant = 1$).

Of course we can view the cumulative allocation rule $\cumalloc$ of
$\qalloc$ as a constraint and observe that $\qalloc$ satisfies the
constraint with equality.  Moreover, among allocation rules that
satisfy $\cumalloc$ as a constraint, $\qalloc$ has the highest
probability on stronger (i.e., lower) quantiles.  Conversely, the
allocation constraint~$\cumcalloc$ (with corresponding
$\calloc(\quant) = \tfrac{\dd}{\dd\quant} \cumcalloc(\quant)$) is met
by any allocation rule~$\qalloc$ that relatively has allocation
probability shifted from stronger quantiles to weaker quantiles.
Specifically, $\calloc$ {\em majorizes} $\qalloc$.

\begin{definition}
\label{def:rev}
We say an allocation rule $\qalloc$ is \emph{weaker} than another allocation
rule~$\calloc$ if $\calloc$ majorizes $\qalloc$ in the sense discussed above.
$\Rev{\calloc}$ is defined for all allocation constraints $\calloc$ as
the revenue of the {\em interim optimal lottery pricing} with
allocation rule weaker than $\calloc$.\footnote{From the agent's
  perspective in a multi-agent mechanism, the allocation constraint
  $\calloc$ is applied at the interim stage of the mechanism, i.e.,
  when the agent knows her own type but considers the types of other
  agents to be drawn from their respective distributions.}
\end{definition}

Recall the ex ante lottery pricing problem of optimally serving the
agent with ex ante probability $\exquant$.  A {\em posted price} is
parameterized by a single price and is a simple example of a lottery
pricing (i.e., one that is deterministic): the two menu items are to
be served and pay the price or not to be served and pay nothing.  The
agent prefers service when her value exceeds the price and, otherwise,
she prefers no service.  For an agent with inverse demand curve
$\Val(\cdot)$, the posted price that serves with probability
$\exquant$ is $\Val(\exquant)$.  It gives expected revenue $\exquant
\, \Val(\exquant)$ which is a lower bound on $\rev(\exquant)$.  Its
allocation rule $\callocstep$ is the reverse step function that is one
on quantiles $[0,\exquant]$ and then zero on $(\exquant,1]$.  This
rule has the most service probability on strong quantiles among all
allocation rules that satisfy the ex ante allocation constraint
$\exquant$.  Of course, the revenue it generates $\exquant \,
\Val(\exquant)$ may not be a concave function of $\exquant$ whereas it
must be that the revenue curve $\rev(\cdot)$ is concave.  It can be
shown, in fact, that $\rev(\cdot)$ is exactly the concave hull of
$\exquant \, \Val(\exquant)$ and the optimal lottery for any
$\exquant$ is given by a posted pricing or, if $\rev(\cdot)$ is linear
at $\exquant$, the convex combination of two posted pricings
(corresponding to the end points of the interval containing $\exquant$
on which $\rev(\cdot)$ is linear).  The allocation rule of this convex
combination is a convex combination of two reverse step functions and,
in the sense described above, relative to posting price
$\Val(\exquant)$ has service probability shifted from stronger
quantiles to weaker quantiles.  This specific form (which is not
obvious) is not important for our rederivation of the optimal
mechanism; what is important is the following proposition (which is
obvious from the above discussion).

\begin{proposition}
\label{prop:optimal=weaker+better}
For single-dimensional linear agents, the ex ante optimal lottery pricings
have weaker allocation rules than posted prices and higher revenue.
\end{proposition}

\paragraph{Revenue Linearity}

We are now ready to give the new derivation of the marginal revenue
mechanism and its optimality.  

\begin{definition}
\label{def:revenue-linear}
An agent is {\em revenue linear} if $\Rev{\cdot}$ is a linear
functional, i.e., if the optimal revenue for allocation constraints
$\calloc = \calloc^A + \calloc^B$ is $\Rev{\calloc} = \Rev{\calloc^A}
+ \Rev{\calloc^B}$.
\end{definition}

We can derive a lower bound on the optimal revenue for any allocation
constraint~$\calloc$ as follows.  The constraint $\calloc$ is a
monotone non-increasing function. As reverse step functions provide a
basis for such functions, we can view $\calloc$ as a convex
combination of reverse step functions.  This convex combination can be
sampled from by drawing $\exquant$ at random from the distribution
$\stepdist{\calloc}$ with density $-\calloc'(\quant) = -
\frac{\dd}{\dd\quant}\calloc(\quant)$ and then posting price
$\Val(\exquant)$ (with allocation rule $\calloc^{\exquant}$).  The
allocation rule of the convex combination is exactly~$\calloc$; its
expected revenue is a lower bound on $\Rev{\calloc}$.

We can derive a better lower bound by, for ex ante constraint
$\exquant \sim \stepdist{\calloc}$, offering the ex ante optimal
lottery pricing (instead of posting price $\Val(\exquant)$).  As the
allocation rule for each of these lottery pricings is weaker than the
corresponding posted pricing allocation rule, the convex combination
of the allocation rules (denote it by $\qalloc$) is weaker than the
allocation constraint $\calloc$.  Therefore, $\qalloc$ is feasible for
$\calloc$ and its revenue gives a lower bound on $\Rev{\calloc}$.
Formally, with $\quant \sim U[0,1]$,
\begin{align*}
\Rev{\calloc} 
&\geq \expect[\exquant\sim\stepdist{\exquant}]{\rev(\exquant)}
\\&\geq \expect[\quant]{-\calloc'(\quant)\,\rev(\quant)}
\\&
= \big[-\calloc(\exquant)\,\rev(\exquant)\big]_0^1 + \expect[\quant]{\rev'(\quant)\,\calloc(\quant)}
\\&
= \expect[\quant]{\rev'(\quant)\,\calloc(\quant)}.
\end{align*}
The second equality follows from integration by parts and the third
equality from $\rev(1) = \rev(0)=0$. (Minor assumption: if the agent
is always served or never served then no revenue is obtained.)  This
construction motivates the following definition.

\begin{definition}
\label{def:marginal-revenue} The {\em marginal revenue} for an agent with quantile 
$\quant$ is $\rev'(\quant) = \frac{\dd}{\dd\quant} \rev(\quant)$; the
marginal revenue for an allocation constraint $\calloc$ is
$\MargRev{\calloc} = \expect[\quant]{\rev'(\quant)\,\calloc(\quant)}$.
\end{definition}

The definition of revenue linearity and the definition of the revenue
curve (as the optimal revenue subject to the ex ante constraint
$\exquant$) immediately imply the following theorem.

\begin{theorem}
\label{t:lin=>rev=margrev}
For a revenue-linear agent, the optimal revenue for an allocation
constraint is equal to its marginal revenue, i.e., for all $\calloc$,
$\Rev{\calloc} = \MargRev{\calloc}$.
\end{theorem}

The revenue linearity of single-dimensional linear agents is a simple
consequence of revenue equivalence \citep{M81} and the fact that the
optimal revenue for ex ante constraint $\exquant$ exceeds the posted
pricing revenue from $\Val(\exquant)$ but has a weaker allocation rule
(\autoref{prop:optimal=weaker+better}).

\begin{theorem}
\label{t:sd=>lin}
An agent with single-dimensional linear utility is revenue linear.
\end{theorem}

\begin{proof}
As we have seen above, the marginal revenue of an allocation
constraint is a lower bound on its optimal revenue. To show revenue
linearity, then, it suffices to upper bound the optimal revenue by the
marginal revenue.

For any allocation rule $\qalloc$ (or constraint) marginal revenue can
be written as
\begin{align}
\label{eq:MR-R}
\MargRev{\qalloc} &= \expect{\smash{-\qalloc'(\quant)\,\rev(\quant)}}\\
\label{eq:MR-R'}
                  &= \expect{\smash{\marg(\quant)\, \qalloc(\quant)}},\ \text{and}\\
\label{eq:MR-R''}
                  &= \marg(1) \cumalloc(1) + \expect{\smash{-\marg'(\quant) \, \cumalloc(\quant)}}.
\end{align}
We already saw the derivation of equation~\eqref{eq:MR-R'}
from~\eqref{eq:MR-R}, which follows from integration by parts and
$\rev(0) = \rev(1) = 0$.  Equation~\eqref{eq:MR-R''} follows from
integrating by parts again and $\cumalloc(0) = 0$ (by definition).
From equation~\eqref{eq:MR-R}, it is apparent that higher revenue
curves give higher revenue (as ``$-\qalloc'(\cdot)$'' is non-negative
for monotone allocation rule $\qalloc(\cdot)$).  From
equation~\eqref{eq:MR-R''}, it is apparent that higher allocation
rules, in the sense of majorization, give higher revenue (as
``$-\marg'(\cdot)$'' is non-negative for concave revenue curve
$\rev(\cdot)$ and majorization requires equality of $\cumalloc(1)$).
 
Let $\rawrev(\exquant)$ denote the expected revenue from posting
price $\Val(\exquant)$, i.e., $\rawrev(\exquant) = \exquant \,
\Val(\exquant)$.  Suppose we optimize for $\calloc$ and get some
(possibly less restrictive) allocation rule $\qalloc$, then optimizing
for $\qalloc$ as a constraint gives the same revenue,
\begin{align*}
\Rev{\calloc} &=\Rev{\qalloc}.
\\ \intertext{By revenue equivalence, the revenue of any allocation
  rule is given by its price-posting revenue curve $\rawrev(\cdot)$.  Therefore,}
\Rev{\qalloc} &= \expect{\smash{-\qalloc'(\quant)\,\rawrev(\quant)}}.
\\ \intertext{As $\rawrev(\quant) \leq \qrev(\quant)$ for all $\quant$, equation~\eqref{eq:MR-R} implies that the marginal revenue from $\rawrev(\cdot)$ is at most that of $\rev(\cdot)$ for allocation rule $\qalloc(\cdot)$:}
\expect{\smash{-\qalloc'(\quant)\,\rawrev(\quant)}} &\leq
\expect{\smash{-\qalloc'(\quant)\,\qrev(\quant)}} = \MargRev{\qalloc}.
\\ \intertext{As $\qalloc$ is majorized by $\calloc$, equation~\eqref{eq:MR-R''} implies that the marginal revenue of $\qalloc(\cdot)$ is at most that of $\calloc(\cdot)$ for revenue curve $\rev(\cdot)$:}
\MargRev{\qalloc} = \marg(1)\, \cumalloc(1) + \expect{\smash{-\marg'(\quant)\,\cumalloc(\quant)}} 
&\leq
\marg(1)\, \cumcalloc(1) + \expect{\smash{-\marg'(\quant)\,\cumcalloc(\quant)}} = \MargRev{\calloc}. \qedhere
\end{align*}
\end{proof}

\begin{corollary}
\label{t:sd-rev=margrev}
For a single-dimensional linear agent, the optimal revenue for an
allocation constraint is equal to its marginal revenue, i.e., for all
$\calloc$, $\Rev{\calloc} = \MargRev{\calloc}$.
\end{corollary}

\paragraph{Multi-agent Mechanisms}

The conclusion of the preceding discussion is that the optimal revenue
for any allocation constraint is equal to its marginal revenue.  

\begin{definition}
Any mechanism and distribution over types induces a profile $\allocs =
(\alloc_1,\ldots,\alloc_n)$ of interim allocation rules.  The {\em
  surplus of marginal revenue} is the sum of the marginal revenues of
interim allocation rules of each agent $\sum_i \MargRev{\alloc_i}$.
\end{definition}

Multi-agent mechanism design problems reduce to single-agent lottery
pricing problems as follows.  The following argument is the standard
in auction theory.  For an agent in the optimal mechanism, her
contribution to the revenue is equal to the marginal revenue of her
allocation rule (\autoref{t:sd-rev=margrev}).  We thus look for the
mechanism that optimizes the surplus of marginal revenue.  Consider
relaxing the incentive constraints (namely: monotonicity of the
allocation rule) and optimizing marginal revenue pointwise.
Specifically, when the agent quantiles are $\quants =
(\quanti[1],\ldots,\quanti[n])$ select the allocation $\allocs =
(\alloci[1],\ldots,\alloci[n])$ to maximize the surplus of
  marginal revenue $\sum_i \margi(\quanti) \, \alloci$ subject to
feasibility of $\allocs$ (e.g., for a single-item auction, serve the
agent with the highest positive marginal revenue, or none if the
marginal revenues are all negative).  Now check that the previously
relaxed incentive constraints are not violated.  Notice that since
revenue curves are concave, the marginal revenues are monotone
non-increasing in quantile, for any agent a stronger (lower) quantile
corresponds to a weakly higher marginal revenue, and so the induced
allocation rule is monotone.  Furthermore, as these allocations
optimize marginal revenue pointwise for all profiles of agent
quantiles, they certainly also maximize marginal revenue in
expectation over the agent quantiles.

Comparing the above construction with the marginal revenue mechanism
framework described in the introduction, the missing
Steps~\ref{step:quantile} and~\ref{step:outcome} are simple.  For
Step~\ref{step:quantile}, the mapping from value to quantile is given
by $\Vali^{-1}(\cdot)$ for each agent $i$ as described above.  For
Step~\ref{step:outcome}, the appropriate payments can be calculated
pointwise as follows:  Agents that are not served pay nothing and an
agent $i$ that is served pays the value $\Vali(\exquanti)$
corresponding to her critical quantile $\exquanti$, i.e., the quantile
after which she would no longer be served (via the payment identity).

\begin{theorem}
\label{t:sd-opt}
The marginal revenue mechanism is revenue optimal for
single-dimensional linear agents.
\end{theorem}

\begin{proof}
The optimal mechanism induces some profile $\allocs$ of interm
allocation rules.  By revenue linearity, the expected revenue of this
profile of interim allocation rules is equal to its surplus of
marginal revenue.  The marginal revenue mechanism selects its outcome
to optimize surplus of marginal revenue pointwise for the feasibility
constraint.  Its expected surplus of marginal revenue is, thus, at
least that of the optimal mechanism.
\end{proof}

\section{Multi-dimensional and Nonlinear Preferences}
\label{sec:general}\label{sec:prelim}

\paragraph{Bayesian mechanism design}

An agent has a private type~$\type$ from type space~$\typespace$ drawn
from distribution $\dist$ with density function $\dens$.  The agent
may be assigned outcome~$\outcome$ from outcome space~$\outcomespace$.
This outcome encodes what kind of service the agent receives and any
payments she must make for the service.  In particular the payment
specified by an outcome $\outcome$ is denoted by $\PAYMENT(\outcome)$.
The agent has a von Neumann--Morgenstern utility function: for
type~$\type$ and deterministic outcome $\outcome$ her utility is
$\util(\type, \outcome)$, and when $\outcome$ is drawn from a
distribution her utility is $\expect[\outcome]{\util(\type,
  \outcome)}$.\footnote{This form of utility function allows for
  encoding of budgets and risk aversion; we do not require
  quasi-linearity.}  We will extend the definition of the utility
function to distributions over outcomes $\distover{\outcomespace}$
linearly.  For a random outcome $\outcome$ from a distribution,
$\PAYMENT(\outcome)$ will denote the expected payment.

\begin{example}[A unit-demand quasi-linear-utility agent]
The preferences of a \emph{unit-demand quasi-linear-utility} agent are
as follows.  There are $m$~\emph{alternatives} and the agent's type is
given by a vector $(\val^1, \ldots, \val^m)$ representing her value
for each alternative.  An outcome is of the form $(\payment,
\pralloc^1, \ldots, \pralloc^m)$, where $\payment$ denotes the
payment, and each $\pralloc^j \in \{0, 1\}$ indicates whether the
agent gets the alternative~$j$, with $\sum_{j} \pralloc^j \leq 1$.
The agent's utility at such an outcome is then given by the linear
form $\sum_{j} \val^j \pralloc^j - \payment$.  When randomizing over
such outcomes, we relax the $\pralloc^j$'s to be in $[0, 1]$, still
with $\sum_{j} \pralloc^j \leq 1$.  Such a distribution with a
price~$\payment$ is called a \emph{lottery}.
\end{example}

\begin{example}[A single-dimensional public-budget agent]
The preferences of a {\em single-dimensional public-budget} agent are
as follows.  The agent has a publicly known budget~$B$, and her
type is given by her private value~$\val$ for an item being auctioned.
An outcome $\outcome = (\alloc, \price)$ indicates by $\alloc \in \{0,
1\}$ whether the agent gets the item, and by $\price$ the amount of
payment she makes.  In contrast to the single-dimensional
linear-utility agents of \autoref{sec:single-dimensional}, this
agent's utility is $\val \cdot \alloc - \price$ only if $\price \leq
B$, and negative infinity otherwise.
\end{example}

There are $n$ agents indexed $\{1,\ldots,n\}$ and each agent $i$ may
have her own distinct type space $\typespace_i$, utility function
$\util_i$, etc.  The agents types are indepently distributed.  A {\em
  direct revelation} mechanism takes as its input a profile of types
$\types = (\type_1, \ldots, \type_n) \in \typespace_1\times
\cdots \times \typespace_n$ and outputs ex post outcome
$\expost{\outcomes}(\types) \in \distover{\outcomespace_1 \times
  \cdots \times \outcomespace_n}$.  
Agent $i$'s {\em ex post outcome rule} is denoted by
$\expost{\outcome}_i(\types)$ and, with the other agents' types drawn
from the distribution, her \emph{interim outcome rule}
$\toutcome_i(\type_i)$ is distributed as
$\expost{\outcome}_i(\type_i,\types_{-i})$ with $\type_j\sim \dist_j$
for each $j\neq i$.  We say that a mechanism is \emph{Bayesian
  incentive compatible} if
\begin{align}
\label{constr:BIC}
\util_i(\type_i,\toutcome_i(\type_i)) \geq \util_i(\type_i,\toutcome_i(\type'_i)), \quad \forall i, \forall \type_i,\type'_i \in \typespace_i.
\tag{BIC}
\end{align}
A mechanism is \emph{interim individually rational} if
\begin{align}
\label{constr:BIR}
\util_i(\type_i,\toutcome_i(\type_i)) \geq 0, \quad \forall i, \forall
\type_i \in \typespace_i.  \tag{IIR}
\end{align}

The mechanism designer seeks to optimize an objective subject to
BIC, IIR, and ex post feasibility.  We
consider the objective of expected revenue, i.e.,
$\expect[\types]{\sum_i \PAYMENT(\toutcome_i (\type_i))}$; however, any
objective that separates linearly across the agents can be considered.
Below we discuss the mechanism's feasibility constraint.

\paragraph{Service constrained environments}
In a {\em service constrained environment} the outcome $\outcome$
provided to an agent is distinguished as being either a {\em service} or a {\em
  non-service} outcome, respectively, with $\ALLOC(\outcome) = 1$ or
$\ALLOC(\outcome) = 0$.  There is a feasibility constraint restricting
the set of agents that may be simultaneously served; there is no
feasibility constraint on how an agent is served.  With respect to the
feasibility constraint any outcome $\outcome \in \outcomespace$ with
$\ALLOC(\outcome) = 1$ is the same.  For example, payments are part of
the outcome but are not constrained by the environment.  An agent may
have multi-dimensional and non-linear preferences over distinct
service and non-service outcomes.

From least rich to most rich, standard service constrained
environments are {\em single-unit environments} where at most one
agent can be served, {\em multi-unit environments} where at most a
fixed number of agents can be served, {\em matroid environments} where
the set of agents served must be an independent set of a given
matroid, {\em downward-closed environments} where the set of agents
served can be specified by an arbitrary set systems for which all
subsets of a feasible set are feasible, and {\em general environments}
where the feasible subsets of agents can be given by an arbitrary set
system that may not even be downward closed.



\paragraph{Ex Ante Lottery Pricings and Revenue Curves}

The only aspect of the marginal revenue approach that translates
identically from single-dimensional preferences to general preferences
is the definition of the \optpricing\@ for allocation probabilities
$\exquant \in [0,1]$.  This is the lottery pricing (i.e., collection
of outcomes where the agent is permitted to choose her type-dependent
favorite) denoted $\toutcomestep(\cdot)$ that optimizes revenue
subject to the constraint that
$\expect[\type]{\ALLOC(\toutcomestep(\type))} = \exquant$.  The
revenue curve for the agent is then given by $\rev(\exquant) =
\expect[\type]{\PAYMENT(\toutcomestep(\type))}$ as per
\autoref{def:revenue-curve}.

\paragraph{Allocation rules}
The first challenge in generalizing the marginal revenue approach to
general preferences is determining the mapping from types to
quantiles.  This challenge arises as there is no explicit ordering of
an agent's type space $\typespace$ by strength.  E.g., if the type is
multi-dimensional then it is unclear which is stronger, a higher value
in one dimension and lower in another or vice versa.  In fact, which
is stronger often depends on the context, e.g., the competition from
other agents.

Our approach is based on two observations.  First, relative to a
mechanism and for a particular agent, the relevant part of the
mechanism is the (interim) outcome rule $\toutcome(\cdot)$. For a
given outcome rule $\toutcome(\cdot)$ an ordering on types by strength
can be defined.  Simply, a type that is more likely to be served is
stronger than a type that is less likely to be served.  I.e., $\type$
is stronger than $\type'$ relative to $\toutcome(\cdot)$ if
$\ALLOC(\toutcome(\type)) \geq \ALLOC(\toutcome(\type'))$.  This
definition induces a mapping from the type space to quantile space;
moreover, the distribution of quantiles induced by this mapping and
the distribution on types is uniform.\footnote{Quantiles are uniformly
  distributed when ties in allocation probability are measure zero;
  when there is a measureable probability of ties, quantiles can be
  defined by drawing uniformly from the interval containing the tie.}
Second, (by the above mapping) any outcome rule $\toutcome(\cdot)$
induces an allocation rule $\qalloc(\cdot)$ that maps quantile to
service probability.  This allocation rule has a simple intuition in
discrete type spaces: For each type $\type \in \typespace$ make a
rectangle of width equal to the probability of the type $\dens(\type)$
and height equal to the service probability of the type
$\ALLOC(\toutcome(\type))$.  Sort the types in decreasing order of
heights; the resulting monotone non-increasing piecewise constant
function from $[0,1]$ to $[0,1]$ is the allocation rule.  This is
generalized for continuous distributions as follows.

\begin{definition}
\label{def:allocation-rule}
For an agent with $\type \in \typespace$ drawn from distribution
$\dist$ and outcome rule $\toutcome(\cdot)$, the {\em allocation rule}
mapping quantiles to service probabilities is given by
$\qalloc(\exquant) = \sup\{ y :
  \prob[\type\sim\dist]{\ALLOC(\toutcome(\type)) \geq y} \leq \exquant\}$.
\end{definition}

\paragraph{Optimal Lottery Pricing}
With the definition of allocation rules for any lottery pricing above,
allocation constrained lottery pricings generalize naturally.  Even
though the order on types may change from one lottery pricing to
another, we can still ask for the lottery pricing with the optimal
revenue subject to a constraint on its allocation rule.  The optimal
lottery pricing for allocation constraint $\calloc$ with cumulative
allocation constraint $\cumcalloc$ is given by the outcome rule
$\toutcome(\cdot)$ that optimizes expected revenue subject to its
corresponding allocation rule $\qalloc$ with cumulative allocation
rule $\cumalloc$ satisfying $\cumalloc(\exquant) \leq
\cumcalloc(\exquant)$ for $\exquant \in [0,1]$ with equality at
$\exquant=1$.  As per \autoref{def:rev} the optimal revenue for
allocation constraint $\calloc$ is denoted $\Rev{\calloc}$.

We will generally denote by $\qalloc$ the optimal allocation rule for
constraint $\calloc$.  The ex ante constraint on total service
probability by $\exquant$ is given by the reverse step function at
$\exquant$ denoted $\callocstep$; the corresponding allocation rule of
the \emph{$\exquant$ \optpricing} is denoted $\qallocstep$.

\paragraph{Revenue Linearity and Marginal Revenue}
Revenue linearity and marginal revenue have the same definitions
(\autoref{def:revenue-linear} and \autoref{def:marginal-revenue}) as
for single-dimensional preferences. The marginal revenue of an
allocation constraint is $\MargRev{\calloc} =
\expect[\quant]{\marg(\quant)\,\calloc(\quant)}$.  By its construction
as the revenue of the appropriate convex combination of
\optpricings\ it is a lower bound on the optimal revenue, i.e.,
$\Rev{\calloc} \geq \MargRev{\calloc}$.  Again by its construction,
revenue linearity would imply that its revenue is equal to the optimal
revenue (\autoref{t:lin=>rev=margrev}).  We will desribe the
marginal-revenue approach for non-revenue-linear agents by analogy to
the single-dimensional case.



\begin{definition}
\label{def:optimal-marginal-revenue}
The {\em single-dimensional analog} of a service constrained
environment for general agents is the environment with
single-dimensional linear agents with the same revenue curves.  The
{\em optimal marginal revenue} for a service constrained environment
for general agents is the optimal revenue of the single-dimensional
analog (which is equal to its surplus of marginal revenue).
\end{definition}

Our approach to multi-agent mechanism design via the single-dimensional
analog is to look at the profile of interim allocation rules induced
by maximization of surplus of marginal revenue and then to construct a
mechanism for general agents that looks to each agent like the convex
combination of the \optpricings\@ for her allocation rule.  For
revenue curves $\rev_1,\ldots,\rev_n$, draw quantiles $\quants =
(\quant_1,\ldots,\quant_n)$ uniformly from $[0,1]^n$, serve to
maximize surplus of marginal revenue pointwise as $\sum_i
\margi(\quanti)\, \alloci$ for feasible $\allocs =
(\alloci[1],\ldots,\alloci[n])$.  We can interpret the allocation
rules induced by this process as allocation constraints for the
general environment and denote them by $\callocsmr =
(\callocmr_1,\ldots,\callocmr_n)$.  As for single-dimensional linear
agents (see \autoref{sec:single-dimensional}), one way to serve an
agent subject to allocation constraint $\calloc$ is to draw a quantile
$\exquant$ from the distribution $\stepdist{\calloc}$ with density
$-\frac{\dd}{\dd\quant}\,\calloc(\quant)$ and run the \optpricing\@
for ex ante constraint $\exquant$.  This approach suggests attempting
to implement the general mechanism with outcome rules that correspond
to allocation rules of the single-dimensional analog.  Denoting the
outcome rule for the $\exquant$ \optpricing\@ for agent $i$ by
$\toutcomestepopt[\exquant]_i(\type_i)$.  The agent's outcome rule
corresponding to constraint $\callocmr_i$ is $\toutcomemr_i(\type_i) =
\int_0^1
\toutcomestepopt[\exquant](\type_i)\,(-\dd\callocmr_i(\quant))$.
There may be multiple ways to implement this profile of outcome rules
ex post; however, the direct approach employed for single-dimensional
linear agents in \autoref{sec:single-dimensional} does not always
generalize.

\begin{definition}
\label{def:mrm}
The {\em marginal revenue outcome rule} of an allocation rule $\alloc$
is $\toutcome(\type) = \int_0^1
\toutcomestepopt[\exquant](\type)\,(-\dd\alloc(\quant))$.  A
\emph{marginal revenue mechanism} is one with interim outcome rules
equal to the marginal revenue outcome rules corresponding to the
optimal marginal revenue.
\end{definition}

\paragraph{Implementation with Revenue Linearity}

We show now that the marginal revenue mechanism generalizes exactly
for general preferences that satisfy revenue linearity. Moreover, we
show that in this case the marginal revenue mechanism inherits all of
the nice properties of the marginal revenue mechanism for
single-dimensional preferences.  Namely, it deterministically selects
the set of agents to serve, it is dominant strategy incentive
compatible (truthful reporting is a best response for any actions of
the other agents), and the mapping from types to quantiles to marginal
revenues is deterministic and {\em context free}\footnote{Note that
  this contrasts with recent algorithmic work in multi-dimensional
  optimal mechanism design where the optimal mechanism is
  characterized by mapping types stochastically to ``virtual values''
  and this mapping is solved for from the feasibility constraint and
  the distributions of all agents types.  See \citet{AFHHM12} and
  \citet{CDW12,CDW12b}.}  in that it does not depend on the
feasibility constraint or other agents in the mechanism.  The
mechanism, however, is optimal among the larger class of randomized
and Bayesian incentive compatible mechanisms.  As motivation for this
result, we will show subsequently that there are multi-dimensional
preferences that are revenue linear, e.g., when multi-dimensional
values are uniformly distributed on a hypercube.

The main challenge of implementing the marginal revenue mechanism is
in specifying Step~\ref{step:quantile}, i.e., the mapping from types
to quantiles, and Step~\ref{step:outcome}, i.e., selecting the
appropriate outcomes for the set of agents that are served.  If,
however, each agent's types are orderable by the following definition,
then both steps are essentially identical to the single-dimensional
case. 

\begin{definition}
\label{def:orderable}
A single-agent problem is {\em orderable} if there is an equivalence
relation on the types, and there is an ordering on the equivalence
classes, such that for any allocation constraint~$\calloc$, the
optimal outcome rule~$\toutcome$ induces an allocation rule that is
greedy by this ordering with ties between types in a same equivalence
class broken uniformly at random.\footnote{By greedy by the given
  ordering, we mean process each equivalance class in order and serve
  the corresponding types with as much probability as possible subject
  to the allocation constraint.  
  If all equivalance classes are measure zero,
  then the resulting allocation rule is equal to the allocation
  constraint.}
\end{definition}

Orderability may look like a stringent and unlikely condition to hold generally.  We note that it holds for
single-dimensional agents and we show now, more generally, that it is a consequence of revenue linearity.

\begin{theorem}
\label{thm:linear=orderable}
For any single-agent problem, revenue linearity implies orderability.
\end{theorem}

The theorem is proved by the following two lemmas which characterize
the structure of optimal lottery pricings.

\begin{lemma}
\label{lem:X=Y}
For a revenue-linear single-agent problem, let $\qalloc$ be the
optimal allocation rule subject to some constraint $\calloc$. Then,
for any $\exquant$ such that $\rev''(\exquant)\neq 0$ we have
$\cumalloc(\exquant) = \cumcalloc(\exquant).$
\label{greedylemma}
\end{lemma}

\begin{proof}
Since $\qalloc$ is the optimal allocation rule subject to $\calloc$, we have
$\Rev{\qalloc} = \Rev{\calloc}$.
Linearity implies that
\begin{equation*}
\MargRev{\calloc} = \int_{0}^1 \qalloc(\quant)\rev'(\quant)\, \dd\quant = \int_{0}^1
\calloc(\quant)\rev'(\quant)\, \dd\quant = \MargRev{\qalloc}.
\end{equation*}

Integrating by parts, we have
\begin{equation}
\inteval{\cumalloc(\quant)\rev'(\quant)}_{0}^1 - \int_{0}^1
\cumalloc(\quant)\rev''(\quant)\,\dd\quant =
\inteval{\cumcalloc(\quant)\rev'(\quant)}_{0}^1- \int_{0}^1
\cumcalloc(\quant)\rev''(\quant)\,\dd\quant.
\label{firsteqingreedy}\end{equation}

Note that $\calloc$ and $\qalloc$ have the same ex ante probability of allocation $\cumcalloc(1) = \cumalloc(1)$;
also by definition $\cumalloc(0) = \cumcalloc(0) = 0$.  Combining these observations with
\eqref{firsteqingreedy} we have
\begin{equation*}
\int_{0}^1 \cumalloc(\quant)\rev''(\quant)\,\dd\quant = \int_{0}^1
\cumcalloc(\quant)\rev''(\quant)\,\dd\quant,
\end{equation*}
and therefore,
\begin{equation}
\int_{0}^1 [\cumalloc(\quant)-\cumcalloc(\quant)]\rev''(\quant)\,\dd\quant = 0. \label{secondeqingreedy}
\end{equation}

Notice that for any $\quant$, $\cumalloc(\quant) - \cumcalloc(\quant)$ and $\rev''(\quant)$ are non-positive
(by domination and concavity, respectively) so their product is non-negative.  Therefore,
\eqref{secondeqingreedy} can be satisfied only if $[\cumalloc(\quant)-\cumcalloc(\quant)]\rev''(\quant)= 0$
for all $\quant$. This implies that if $\rev''(\quant)<0$, then we must have $\cumalloc(\quant) =
\cumcalloc(\quant)$, which completes the proof.
\end{proof}

\autoref{lem:X=Y} in particular implies that for $\exquant$ with
$\rev''(\exquant) \neq 0$ the $\exquant$ \optpricing\ 
(i.e., with allocation constraint given by the reverse step function
$\callocstep$) has allocation rule $\qallocstep = \callocstep$.  I.e.,
the $\exquant$ \optpricing\ has only full lotteries (all types are
served with either probability one or zero).

For any such $\exquant$, define $\typespace_{\exquant}$ to be the set of
types allocated (with full lotteries) in the optimal allocation
subject to $\callocstep$. The following lemma shows that these sets
are nested.

\begin{lemma}
\label{l:nested}
For a revenue-linear single-agent problem, for any $\exquant_1>\exquant_2$ and
$\rev''(\exquant_1),\rev''(\exquant_2)\neq 0$, we must have
$\typespace_{\exquant_1} \supseteq \typespace_{\exquant_2}$.
\end{lemma}

\begin{proof}
Assume for contradiction that $\typespace_{\exquant_2}\backslash \typespace_{\exquant_1} \neq \emptyset$. Let
$\alpha = \dist(\typespace_{\exquant_2}\backslash \typespace_{\exquant_1})>0$. Consider the following allocation
constraint
\begin{align*}
    \calloc(\quant) &=
        \begin{cases}
            1 & \quant\leq \exquant_2 \\
            1/2 & \exquant_2 < \quant \leq \exquant_1 \\
            0 & \exquant_1 < \quant.
        \end{cases}
\end{align*}
By revenue linearity, the revenue of the optimal auction subject to $\calloc$ is
$[\rev(\exquant_1)+\rev(\exquant_2)]/2$. Notice that the mechanism that runs
$\rev(\exquant_1)$ and $\rev(\exquant_2)$ each with probability 1/2 achieves this revenue. The allocation rule $\qalloc$ of this
mechanism is
\begin{align*}
    \qalloc(\quant) &=
        \begin{cases}
            1 & \quant\leq \quant_2-\alpha \\
            1/2 & \quant_2-\alpha \leq \quant \leq \quant_1 + \alpha \\
            0 & \quant_1 + \alpha \leq \quant.
        \end{cases}
\end{align*}
Notice that this allocation rule is dominated by $\calloc$, and achieves the optimal revenue. Yet, we have
\begin{align*}
\cumcalloc(\exquant_1) = \int_{\quant = 0}^{\exquant_1} \calloc(\quant)\:\dd\quant > \int_{\quant=0}^{\exquant_1} \qalloc(\quant)\:\dd\quant = \cumalloc(\exquant_1).
\end{align*}
This contradicts \autoref{greedylemma}.
\end{proof}

\begin{proof}[\NOTSTOC{Proof }of \autoref{thm:linear=orderable}]
By \autoref{l:nested}, all $\exquant$ \optpricings\ order
the types by the same equivalance classes.  By revenue linearity the
optimal lottery pricing for an allocation constraint $\calloc$ is a
convex combination of the $\exquant$ \optpricings.
Therefore, it allocates greedily to types by the same equivalance
classes.
\end{proof}

Given orderability and the fact that (by \autoref{greedylemma}) the
optimal $\exquant$ \optpricings\ are full lotteries for
$\exquant$ for which $\rev(\exquant)$ is locally non-linear, the marginal
revenue mechanism is easy to define.

\begin{definition}
\label{def:context-free-marg-rev-mech}
The {\em marginal revenue mechanism for orderable agents} works as follows.
\begin{enumerate}
\item \label{step:map-type-to-quantile} Map reported types $\types =
  (\type_1,\ldots,\type_n)$ of agents to quantiles $\quants =
  (\quant_1,\ldots,\quant_n)$ via the implied ordering.\footnote{This
  ordering can be found by calculating the optimal single-agent
  mechanism for allocation constraint $\calloc(\quant) = 1-\quant$.}
\item \label{step:max-marg-rev} Calculate the marginal revenue of each
  agent $i$ as $\rev'_i(\quant_i)$.
\item \label{step:critical-quantiles} For each agent $i$, calculate
  the maximum quantile $\exquanti$ that she could possess and be in
  the marginal revenue maximizing feasible set (breaking ties
  consistently).
\item Offer each agent $i$ the $\exquanti$ \optpricing.
\end{enumerate}
\end{definition}

\begin{proposition}
The marginal revenue mechanism deterministically selects a feasible
set of agents to serve and is dominant strategy incentive compatible.
\end{proposition}

\begin{proof}
Because ties are broken consistently, critical values cannot fall in intervals where the revenue curve is
locally linear (and the marginal revenue curve is locally constant).  Therefore, the lottery pricings offered
to each agent are full lotteries; each type is deterministically served or not served.  Feasibility follows
as the set of agents that select service outcomes is exactly the marginal revenue maximizing set subject to
feasibility.  To verify the dominant strategy incentive compatibility consider any agent $i$'s perspective.
The parameter $\exquanti$ is a function only of the other agents' reports; the agent's outcome is determined
by the $\exquanti$ \optpricing\ which is incentive compatible for any $\exquanti$.
\end{proof}

\begin{proposition}
In service constrained environment with revenue-linear agents, the
marginal revenue mechanism obtains the optimal marginal revenue (which
equals the optimal revenue).
\end{proposition}

\paragraph{Testing Revenue Linearity}

Revenue linearity is computationally easy to test.  From the concavity
of $\Rev{\cdot}$ and equality of revenue and marginal revenue for
allocation constraints $\callocstep$ which are a basis for general
allocation constraints, it suffices to check the equality of revenue
and marginal revenue, i.e., $\Rev{\calloc} = \MargRev{\calloc}$, for
any allocation constraint $\calloc$ with positive derivative
($\calloc$ as a convex combination of $\callocstep$ has positive
density on each $\exquant$).  For example, $\calloc(\quant) =
1-\quant$ is such an allocation constraint.  Since the theorem
facilitates testing the property, we discretize the quantile space to
$Q_N = \{0, \tfrac 1 N \cdots, \tfrac {N - 1}{N}, 1\}$ for an
arbitrary integer $N > 0$.

\begin{theorem}
\label{thm:test-rev-lin}
Let $\calloc: Q_N \to [0, 1]$ be any strictly decreasing function, (e.g.,
$\calloc(\exquant) = 1 - \exquant$).  Then a
set of single-agent pricings are revenue linear (or, more precisely,
$\Rev{\cdot}$ is a
linear functional for non-increasing functions mapping $Q_N$ to $[0, 1]$), if $\Rev{\calloc} = \MargRev{\calloc}$.
\end{theorem}

\begin{proof}  Consider the $N + 2$ reverse step functions that ``steps
down'' from $1$ to~$0$ at a point in $Q_N$.  Any non-increasing
function mapping $Q_N$ to $[0, 1]$ is a convex combination of these
base functions, and a strictly decreasing function can be written
uniquely as such a convex combination.  Therefore $\Rev{\calloc} =
\MargRev{\calloc}$ amounts to saying that $\Rev{\cdot}$ is linear on
one interior point in a simplex, and the theorem states that
$\Rev{\cdot}$ is linear on the whole simplex.  If we shift
$\Rev{\cdot}$ by a linear functional such that it is zero on all the
base functions, then this theorem follows from the simple fact that,
if a concave function~$g$ is $0$ on all vertices of a simplex
\emph{and} one interior point~$A$, then $g$ is uniformly $0$ on the
simplex.  To see this, suppose on point~$B$ in the simplex, $g(B) \neq
0$.  By concavity, $g(B) > 0$.  $A$ can be written as a convex
combination of $B$ and vertices of the simplex with a strictly
positive coefficient on~$B$.  (E.g., connect $B$ and $A$ with a
straight line and extend it to intersect at one facet of the simplex
formed by $N-1$ vertices, then $A$ can be written as a convex
combination of $B$ and these $N-1$ vertices, where the coefficient
on~$B$ in the decomposition is strictly positive.)  But the concavity
of $g$ implies $g(A) > 0$, a contradiction.
\end{proof}

\begin{example}[A multi-dimensional revenue-linear example]
The example of the seller who can paint her car red or blue as she sells it to agents with independent and
uniform values for each color is revenue linear (proof given in
\autoref{sec:hypercube-linear}). Therefore, the
marginal revenue mechanism is optimal and its simple form can be derived from
\autoref{def:context-free-marg-rev-mech} as follows.  For a unit-demand agent with values for $m$ variants
of a service (i.e., possible colors of the car) distributed uniformly on $[0,1]^m$, we show that the ex ante optimal mechanism for constraint $\exquant$ is to post a price of $\sqrt[m]{1-\exquant}$ for any service. Notice that such a price will be accepted with probability $\exquant$, and therefore the revenue function is $\rev(\exquant) = \exquant\sqrt[m]{1-\exquant}$, and the marginal revenue function is $\rev'(\exquant) =
(1-\exquant)^{1/m-1}(1-\exquant - \exquant / m)$. The quantile of each type is
$\type=(\type^1,\ldots,\type^m)$ to be $\quant = 1-(\max_i\type^i)^m$. Notice
that both the mapping and the marginal revenue function are monotone. Therefore serving the agent with the
highest marginal revenue (\autoref{def:context-free-marg-rev-mech}) means serving the player with the highest
value for any kind of service and charging her the minimum she needs to bid to exceed the second-highest
value (subject to the reserve of $\sqrt[m]{\frac{1}{m+1}}$ which is where the marginal revenue becomes zero). Revenue-linearity implies that this mechanism is optimal.

\end{example}

\section{Implementation}
\label{sec:implementation}

The marginal revenue mechanism for agents with orderable types
(\autoref{def:context-free-marg-rev-mech}) does not extend to general
agents.  In this section we give two approaches for defining the
marginal revenue mechanism more generally.  The first approach assumes
that the parameterized family of $\exquant$ \optpricings\ satisfies a
natural monotonicity requirement: that the probability that an agent
with a given type is served is monotone in the ex ante
constraint~$\exquant$.  Key to this construction is a
\emph{randomized} mapping from an agent's types to quantiles that is
determined by the agent's type space and distribution alone, and is
therefore context free, i.e., unaffected by the presence of other
agents and the feasibility constraints.  Consequently, (a) the
resulting mechanism is dominant strategy incentive compatible but, (b)
the set of winners is generally a randomized function of the profile
of types.  The second approach is brute-force but easily computable
and completely general.  It results in a Bayesian incentive compatible
mechanism.  Both these mechanisms will differ from the marginal
revenue mechanism for orderable types only in the first (mapping types
to quantiles) and last (serving each agent if her quantile is at most
her critical quantile) steps; these changes can be mix-and-matched for
different agents in the same mechanism.


We conclude this section by describing a relevant class of agents for
which the \optpricings\ satisfy the monotonicity property required by
the first approach.  The example considers single-dimensional agents
with a public budget that constrains their maximum payment.

\subsection{Monotone \optpricings}
\label{sec:monotone}


We consider agents whose \optpricings\ satisfy the following natural
monotonicity property.

\begin{definition} 
\label{def:fu-monotone}
An agent has {\em monotone \optpricings} if, given her type, the
probability she wins in the $\exquant$ \optpricing\ is
monotone non-decreasing in $\exquant$.
\end{definition}

Suppose that the $\exquant$ \optpricing\ for an agent
each consists of a menu of full lotteries.  I.e., for any type of the
agent she will choose a lottery that either serves her with
probability one or zero.  In this case the monotone \optpricings\ assumption would require that the sets of types served for each
$\exquant$ are nested.  There is a simple deterministic mapping from
types to quantiles in this case: set the quantile of a type to be the
minimum $\exquant$ such that the $\exquant$ \optpricing\ serves the
type.  Below, we generalize this selection procedure to the case of
partial lotteries (where types may be probabilistically served).

Recall that the $\exquant$ \optpricing, as a function of the agent's
type, has an allocation and outcome rule $\allocstepopt[\exquant]$ and
$\toutcomestepopt[\exquant]$, respectively.  Fix the type of the agent
as $\type$ and consider the function $G_\type(\exquant) =
\allocstepopt[\exquant](\type)$ which, by the monotonicity condition
above, can be interpreted as a cumulative distribution function.
Recall that $\exquant$ \optpricing\ has probability of service
$\expect[\type]{\allocstepopt[\exquant](\type)} = \exquant$.
Therefore, if $\type$ is drawn from the type distribution and then
$\quant$ drawn from $G_\type$ then the distribution of $\quant$ is
uniform on $[0,1]$.

\begin{lemma} 
\label{lem:uniform-dist}
If $\type \sim \dist$ and $\quant \sim G_\type$ then $\quant$ is $U[0,1]$.
\end{lemma}

\begin{definition} 
\label{def:monotone-marg-rev-mech}
The {\em marginal revenue mechanism for agents with monotone \optpricings} works as follows.
\begin{enumerate}
\item 
\label{step:quantile-monotone}
Map reported types $\types = (\type_1,\ldots,\type_n)$ of agents
  to quantiles $\quants = (\quant_1,\ldots,\quant_n)$ by sampling
  $\quant_i$ from the distribution with cumulative distribution
  function $G_{\type_i}(\quant) = \allocstepopt_i(\type_i)$.
\item Calculate the marginal revenue of each agent $i$ as $\rev'_i(\quant_i)$.
\item For each agent $i$, calculate the maximum quantile $\exquant_i$
  that she could possess and be in the marginal revenue maximizing
  feasible set (breaking ties consistently).
\item
\label{step:outcome-monotone}
For each agent $i$, offer the $\exquant_i$ \optpricing\
conditioned so that $i$ is served if $\quant_i \leq \exquant_i$ and not served
otherwise.
\end{enumerate}
\end{definition}

The last step of the marginal revenue mechanism warrants an
explanation.  In the $\exquant_i$ \optpricing, the outcome that $i$
would obtain with type $\type_i$ may be a partial lottery, i.e., it
may probabilistically serve $i$ or not.  The probability that $i$ is
served is $\allocstepopt[\exquant_i]_i(\type_i) = 
\prob[\quant_i]{\quant_i \leq \exquant_i} = G_{\type_i}(\exquant_i)$ by our choice of
$\quant_i$.  When we offer agent $i$ the $\exquant_i$
\optpricing\ we must draw an outcome from the distribution given by
$\toutcomestepopt[\exquant_i]_i(\type_i)$.  Some of these outcomes
are service outcomes, some of these are non-service outcomes.  If
$\quant_i \leq \exquant_i$ then we draw an outcome from the
distribution $\toutcomestepopt[\exquant_i]_i(\type_i)$ conditioned
on service; if $\quant_i > \exquant_i$ then we draw an outcome
conditioned on no-service.  Notice that, while it may not be feasible
to serve all agents who receive non-trivial partial lottery, this
method coordinates across the partial lotteries which agents to serve
to maintain the right distribution on agent outcomes and ensure
feasibility.

\begin{proposition}
The marginal revenue mechanism for agents with monotone step
mechanisms is feasible and dominant strategy incentive compatible.
\end{proposition}

\begin{proof} 
Feasibility follows as the set of agents that select service outcomes
is exactly the marginal revenue maximizing set subject to feasibility.
To verify the dominant strategy incentive compatibility consider any
agent $i$'s perspective.  The parameter $\exquant_i$ is a
(randomized) function only of the other agents' reports; the agent's
outcome is determined by the $\exquant_i$ \optpricing\ which is
incentive compatible for any $\exquant_i$.
\end{proof}

\begin{theorem} 
The marginal revenue mechanism for agents with monotone \optpricings\ implements
marginal revenue maximization (\autoref{def:mrm}).
\end{theorem}

\begin{proof} 
From each agent $i$'s perspective, the other agents' quantiles are distributed
independently and uniformly on $[0,1]$ (\autoref{lem:uniform-dist}).  Therefore,
this agent faces a distribution over \optpricings\ that is identical to
the distribution of ``critical quantiles'' in the maximization of
marginal revenue, i.e., with density $\tfrac{\dd}{\dd\exquant}\distmr_i(\exquant)$.
\end{proof}

\subsection{General \optpricings}
\label{sec:nonmonotone}

For general agents for whom the \optpricings\ do not satisfy the
monotonicity condition (\autoref{def:fu-monotone}), we give in
\autoref{sec:nonmonotone-app} an simple procedure to implement the
marginal revenue mechanism (recall \autoref{def:mrm}).  This mechanism
is given by \autoref{def:general-marg-rev-mech} in
\autoref{sec:nonmonotone-app}.  The key to the proof of
\autoref{thm:nonmonotone-imp} is a variation of the technique of
vector majorization \citep{HLP-29}.
\begin{theorem} 
\label{thm:nonmonotone-imp}
For service constrained environments, there is a simple Bayesian
incentive compatible implmentation of the marginal revenue mechanism.
\end{theorem}

\IFSTOCELSE
{\section{Example: single-dimensional a\-gent with public budgets}}
{\subsection{Example: single-dimensional agents with public budgets}}
\label{sec:budget}

In this section we exhibit a class of single-agent problems with
non-linear utilities that has monotone \optpricings\
(\autoref{def:fu-monotone}).  Consider an agent with a
single-dimensional value for receiving a good but has a public budget
that limits the payment she could make.  Her utility is her value for
receiving the good minus her payment as long as her payment is at most
her budget.  We show that under standard conditions on the agent's
valuation distribution, this single agent problem has monotone
\optpricings.

The following proposition is a consequence of techniques developed by \citet{LR96} and \citet{PV08}; for completeness we provide a proof in \autoref{sec:budget-app} whose steps largely resemble the ones in these two references.

\begin{proposition}
\label{prop:budget}
For regular distribution $\dist$ with non-decreasing density, budget
$B$, and $\exquant \leq 1-\dist(B)$, the $\exquant$ \optpricing\ offers a single take-it-or-leave-it lottery for price $B$
that serves with probability $\pralloc$, where $\pralloc$ is the solution to the
equation $\exquant = \pralloc[1 - \dist(B / \pralloc)]$.  This lottery is bought by the agent when her value is at least $B/\pralloc$ which happens with
probability $1-\dist(B/\pralloc)$.
\end{proposition}  

For $\exquant > 1 - \dist(B)$, it is easy to see that the budget does not bind and the $\exquant$ \optpricing\ is the same as when there is no budget.

Notice that the allocation rule of the mechanism satisfying
\autoref{prop:budget} is a function that steps from 0 to $\pralloc$ at
value $B/\pralloc$.  The required payment $B$ can be viewed as ``the
area above the allocation curve'' which is given by a rectangle with
width $B/\pralloc$ and height $\pralloc$.  If $\pralloc$ increases,
$B/\pralloc$ decreases and more types are served and with a higher
probability; thus, the ex ante probability of service is increased.
Analogously, if we increase the ex ante probability of service, we
enlarge the set of types served and their probability of service.  We
conclude with the following consequence.

\begin{theorem}
\label{thm:budget-main}
An agent with value drawn from a regular distribution with
non-increasing density has monotone \optpricings.
\end{theorem}

\begin{proof}
The only case not argued by the text above is when $\exquant \geq
1-\dist(B)$.  In this case, the budget is not binding and the \optpricing\ posts price $p$ that
satisfies $\exquant =
1-\dist(p)$ and serves agents willing to pay this price with
probability one.  The \optpricings\ are monotone over these
quantiles as well.
\end{proof}

\begin{example}[Implementation with public budgets]
\label{ex:budget}
The following procedure implements the marginal revenue mechanism in a single item auction for bidders each with a publicly
known budget $B$ and value drawn uniformly from $[0, 1]$.  The auction is easier
to describe separately for the two cases when $B < 1/2$ and $B \geq 1/2$.  

For $B < 1/2$, if no bidder bids above~$B$, the item is not sold; if only one
bidder bids above~$B$, she wins the item and makes a payment of~$B$; if at least
two bidders bid above~$B$, the winner of the item will be decided among these
bidders by a random procedure described shortly.  The winner always makes a payment of~$B$.

For $B \geq 1/2$, if no bidder bids above~$B$, a second price auction is run
with a reserve price of~$1/2$; if one bidder bids above~$B$, she wins the item
and makes the same payment as in a second price auction with reserve~$1/2$; if
at least two bidders bid above~$B$, one of them is decided to be the winner by
a random procedure, but all bidders that bid above a randomly chosen threshold
also makes a payment of~$B$.

Now we describe the random procedure used to determine the winner in both cases.
Note again that only bidders who bid at least~$B$ will enter this procedure.
Each such bidder~$i$ draws a random number~$r_i$ uniformly from $[0, 1]$, and
her quantile $\quant_i$ will be $\max \{r_i,  B / \vali\} - B$.  Whichever
bidder $i^*$ having the smallest quantile is declared the winner.  The threshold
above which other bidders make the payment is $B / (B + \quant_{i^*})$.


\autoref{sec:budget-app}
gives the derivation showing that this is the instantiation of
\autoref{def:monotone-marg-rev-mech}.  

Note the role played by the random mapping in this example.  When multiple
bidders bid above~$B$, the highest bidder is not guaranteed to win the item.
Her higher value helps her obtain a lower quantile by posing a smaller $B /
\vali$, but with positive probability she may lose to a lower bidder.

\end{example}

\section{Approximation}
\label{sec:approx}

In previous sections, we have shown that for any service constrained
environment the marginal revenue mechanism can be implemented.  In
\autoref{sec:general} we have also shown that for revenue linear
agents, it obtains the optimal revenue.  In this section, we show
that, quite generally, the optimal marginal revenue is a good
approximation to the optimal revenue.

We will give two approaches for approximation bounds.  The first kind
of bound is based on the single-agent problem, i.e., the distribution
and type space of each agent: if for all allocation
constraints~$\calloc$, the marginal revenue $\MargRev{\calloc}$ is a
good approximation to the optimal revenue $\Rev{\calloc}$, then the
marginal revenue mechanism is a good approximation to the optimal
mechanism.  The second approach will derive approximation bounds from
the feasibility constraint.  With no feasibility constraint, marginal
revenue maximization is optimal; for matroid environments, it remains
a $1 - 1/e$ approximation; and for general downward-closed
environments with $n$ quasi-linear-preference agents, it gives an $O(\log
n)$ approximation.

Of course, if we are in an environment where agent-based arguments
imply an $\alpha$ approximation and feasibility-based arguments
imply a $\beta$ approximation, the marginal revenue mechanism is in
fact a $\min(\alpha,\beta)$ approximation.  For revenue linear agents,
$\alpha = 1$ (and the optimal marginal revenue gives the optimal revenue); the
approximation smoothly degrades in $\alpha$ as the environment becomes
less revenue linear until it reaches the approximation bound $\beta$
given by the feasibility constraint.

\subsection{Agent-based Approximation}
\label{sec:agent}

If, for all allocation constraints, the marginal revenue is close to
the optimal revenue, then marginal revenue maximization is
approximately optimal.  One approach to deriving such a bound is to
give a linear upper bound on the optimal revenue and a lower bound
through a class of, what we refer to as, \pseudopricings.  An
\pseudopricing\ respects an ex ante constraint but may not be optimal.
If for every ex ante service probability $\exquant$ the $\exquant$
\pseudopricing\ approximates the linear upper bound, then for all
allocation constraints~$\calloc$, the marginal revenue
$\MargRev{\calloc}$ approximates the optimal revenue $\Rev{\calloc}$.
Furthermore, these \pseudopricings\ can be directly optimized over and
the same approximation factor is obtained.  Such an approach might be
desirable if the \pseudopricings\ are better behaved than the
(optimal) \optpricings, e.g., if they are easy to compute, respect an
ordering on types (\`a la \autoref{def:orderable}), or are monotone
(\`a la \autoref{def:fu-monotone}).
This approach is formalized by the following sequence of definitions
and propositions.

\begin{proposition} 
\label{prop:approx-linear}
If for any agent~$i$ and allocation constraint $\calloc_i$, the
marginal revenue $\MargRev{\calloc_i}$ is at least an $\alpha$
approximation to the optimal revenue $\Rev{\calloc_i}$, then the
marginal revenue mechanism in the multi-agent setting is an $\alpha$
approximation to the optimal mechanism.
\end{proposition}

\begin{definition}
\label{def:linear-revenue-bound}
An {\em linear revenue bound}, $\UB$, is a function mapping an allocation constraint to a revenue, which is
\begin{enumerate}
\item linear in the allocation constraint, i.e., for all allocation
constraints $\calloc = \calloc^A + \calloc^B$, $\UB(\calloc) = \UB(\calloc^A) +
\UB(\calloc^B)$; and 
\item an upper bound on revenue for all allocation constraints, i.e.,
  $\forall \calloc,\ \UB(\calloc) \geq \Rev{\calloc}$, and 
\end{enumerate}
\end{definition}

\begin{definition} 
\label{def:pseudo-step-mechanism}
A {\em \pseudopricing} is one that respects an ex ante service probability
constraint but is not necessarily revenue optimal for such a constraint.  The
revenue of a $\exquant$ \pseudopricing\ 
is denoted $\pseudorev(\exquant)$; and the {\em pseudo marginal revenue}
for allocation constraint $\calloc$ is $\PSEUDOMargRev{\calloc} =
\expect{\pseudorev'(\exquant)\calloc(\exquant)}$.
\end{definition}

We can assume without loss of generality that the pseudo marginal
revenue $\pseudorev$ is concave.  If it is not we could always
redefine the class by taking its closure with respect to convex
combination and letting the $\exquant$ \pseudopricing\ be the
revenue-optimal lottery pricing in the closure that serves with ex
ante probability $\exquant$.  This construction is analogous to the
ironing method of \citet{M81}.

\begin{proposition}
\label{prop:pseudo-approx}
For a given linear revenue bound $\UB$, if for all $\exquant \in [0,
  1]$ the $\exquant$ \pseudopricing\ $\alpha$ approximates the
bound on the $\exquant$ ex ante constrained revenue
$\UB(\callocstep)$, then the pseudo marginal revenue
$\alpha$ approximates the optimal revenue for all allocation
constraints.
\end{proposition}

\begin{proof} 
This proposition follows from linearity of both the revenue bound and
pseudo marginal revenue.
\end{proof}

\begin{definition} 
The {\em pseudo marginal revenue mechanism} is the one that maximizes
pseudo marginal revenue via any of the approaches of
\autoref{def:context-free-marg-rev-mech},
\autoref{def:monotone-marg-rev-mech}, or
\autoref{def:general-marg-rev-mech} that applies.
\end{definition}

\paragraph{Pseudo \optpricing\ for downward-closed unit-demand agents}
\label{sec:unit-demand} 


We illustrate the methodology proposed above for the example of
downward-closed service-constrained environments and unit-demand
agents.  Recall for unit-demand agents a service outcome is one of $m$
alternatives.  An agent's type is described by the vector
$(\val^1, \ldots, \val^m)$, her valuations for each of the $m$
alternatives, and her utility for obtaining alternative~$j$ with
payment~$\payment$ is simply $\val^j - \payment$. The agent's type is
drawn from a product distribution over the distinct alternatives.

\citet{CHMS10,CMS-10} show, for a single-unit demand agent with values for distinct alternatives from a product distribution and no feasibility constraint, that
  individually pricing alternatives is a four approximation to optimal
  lottery pricing.  Our approach in this section will be to extend
  this result to settings with ex ante and interim allocation
  constraints.  Our generalization preserves the approximation bound
  of four and exposes the approximate linearity condition required by
\autoref{prop:approx-linear}.


Consider the syntactically-related problem of selling a single
item to one of $m$ single-dimensional agents with values drawn
from a product distribution, i.e., the value $\val_i$ of agent $i$ is
drawn independently from $\dist_i$.  As described earlier
(\autoref{sec:single-dimensional}), the optimal auction for this
single-dimensional problem is well understood.  Agent values are
mapped to virtual values (equivalent to each agent's marginal
revenue), and the agent with the highest positive virtual value is
selected as the winner of the auction.  We refer to this auction
environment as the single-dimensional {\em representative
environment}, the revenue obtained by the optimal auction as the {\em
optimal representative revenue,} and the agents participating in the
auction as {\em representatives}.

Notice that if these representatives were all colluding together the
problem would be identical to our original single-agent unit-demand
problem where the alternatives correspond to the identity of the
winning representative.  We refer to this original environment as the
{\em unit-demand environment} and the revenue of the optimal lottery
pricing as the {\em optimal unit-demand revenue}.  \citet{CHMS10,
CMS-10} considered quantifying the performance of optimal unit-demand
lottery pricings relative to the optimal representative revenue.  The
approach of \citet{CHMS10} is to set individual prices for each
alternative in the unit-demand environment so as to mimic the outcome
of the optimal auction for the representative environment.  As the
optimal auction in the representative environment orders
representatives by virtual values, a natural approach to pricing the
alternatives in the unit-demand environment is to set a uniform
virtual price, i.e., the price for each alternative has the same
virtual value (with respect to the distribution from which the agent's
value for that alternative is drawn).\footnote{As mentioned above, a
representative's virtual value is equal to their marginal revenue.
For clarity of discussion and to disambiguate the marginal revenue of
the unit demand agent versus that of his representatives we will refer
to a representative's marginal revenue as his virtual value.}  The
prices in value space are generally distinct when the agent's value
distributions for the alternatives are non-identical.  \citet{CHMS10}
show that the unit-demand revenue of such a pricing is a
2-approximation to the optimal representative revenue; \citet{CMS-10}
show that the optimal unit-demand revenue (e.g., from lottery
pricings) is at most twice the optimal representative revenue.
Combining these two results, {\em uniform virtual pricing} is a
4-approximation to the optimal unit-demand revenue.

We generalize the approach above to the single-agent problem of
serving an agent with independent values for $m$ alternatives subject
to an allocation constraint $\calloc$.  In particular, twice the
optimal representative revenue is a linear revenue bound
(\autoref{def:linear-revenue-bound}), 
and for any allocation constraint it upper
bounds the optimal (unit-demand) revenue.  We define a class
of \pseudopricings\ where the $\exquant$ \pseudopricing\ is given by a
uniform virtual pricing that sells with probability $\exquant$.  Since
the virtual values are weakly increasing in the representative agents'
values, the sets of types served by these \pseudopricings\ respect an
ordering on types (\autoref{def:orderable}).  Therefore, the pseudo
marginal revenue mechanism can be implemented via the marginal revenue
mechanism for orderable agents
(\autoref{def:context-free-marg-rev-mech}).  Finally, we show that for
all $\exquant$ the $\exquant$ \pseudopricing\ is a four approximation
to the linear upper bound given by twice the optimal representative
revenue.  This result, with \autoref{prop:pseudo-approx}, implies that
the pseudo marginal revenue mechanism is a four approximation to the
optimal revenue for any allocation constraint.  The proof
of \autoref{thm:unit-demand}, below, is a non-trivial but
straightforward extension of \citet{CHMS10,CMS-10} and we include it
in \autoref{app:unit-demand}.

\begin{definition} 
\label{def:unit-demand-pseudo-step}
The $\exquant$ \pseudopricing\ for a unit-demand agent with values for
alternatives drawn independently from $\dist^1,\ldots,\dist^m$ is
given by the pricing that sets a uniform virtual price for the
alternatives such that the probability that the agent buys any
alternative is equal to $\exquant$.  (If this class does not have a
concave pseudo revenue curve we take its closure with respect to
convex combination to make it concave; if this class does not have a
monotone non-decreasing pseudo revenue curve $\pseudorev(\cdot)$ we
invoke downward closure to make it monotone.)\footnote{For showing
approximate linearity for downward-closed environments it is expedient
to incorporate the downward closure into the outcome space by
duplicating each non-service outcome and relabeling the duplicate
outcome as a service outcome.  This transformation is allowed for
downward closed environments because we are always allowed to withhold
service to an agent who would otherwise be served and this withholding
will not violate the feasibility constraint.  Of course, with such a
transformation the revenue curves are non-decreasing.}
\end{definition}

\begin{theorem} 
\label{thm:unit-demand}
In downward-closed (service constrained) environments with unit-demand
agents, both the pseudo marginal revenue mechanism and the marginal
revenue mechanism give 4-approximations to the optimal revenue.
\end{theorem}

\subsection{Feasibility-based Approximation}

We now show that feasibility constraints imply approximation bounds.
As a first trivial observation, if there is no feasibility constraint
(e.g., for digital good environments) then marginal revenue
maximization is optimal.  With no feasibility constraint, each agent
can be considered separately.  For any agent $i$, suppose the revenue
optimal mechanism serves with probability~$\exquanti$, by definition
the revenue it obtains is equal to that of the $\exquanti$
\optpricing.  The optimal revenue $\sum_i\Rev{\callocstepi}$ is equal
to the marginal revenue $\sum_i \MargRev{\callocstepi} = \sum_i
\revi(\exquanti)$.  This observation approximately generalizes as
follows.  The marginal revenue mechanism is a an $e/(e - 1)$
approximation in service-constraind matroid environments and 
an $O(\log n)$ bound for downward-closed environments 
on $n$ quasi-linear-utility agents.

\paragraph{Matroid environments, by single-dimensional-agent reduction}
\label{sec:matroid}

Marginal revenue maximization is an $e/(e-1)$ approximation for
service-constrained matroid environments, i.e., when the feasibility
constraint is induced by independent sets of a matroid set system.
Multi-unit environments, where at most a fixed number $k$ of the
agents can be simultaneously served, are a special case of matroid
environments (corresponding to the $k$-uniform matroid).  For the k=1
unit environment, which corresponds to a single-item auction, the
bound remains $e/(e-1)$; for general $k$ the bound improves to
$\sqrt{2\pi k}/(\sqrt{2 \pi k} -1)$, as simplified by Stirling's
approximation.  These results follow by reduction to the correlation
gap theorem of \citet{Y11}.

Our approach is to reduce the question of approximation of the optimal
mechanism by the marginal revenue mechanism to a question of
approximation in the single-dimensional analog environment
(recall \autoref{def:optimal-marginal-revenue}).  In particular, we
consider relaxing the feasibility constraint to hold \emph{ex ante} instead
of ex post.  Such a relaxation potentially enables a higher revenue to
be obtained.  The single-dimensional-agent approximation question is
to quantify the extent to which the optimal mechanism for the ex post
feasibility constraint approximates the optimal mechanism for the ex
ante feasibility constraint.


\begin{definition}
\label{def:ex-ante-optimal-mechanism}
A profile of ex ante service probabilities $\exquants =
(\exquanti[1],\ldots,\exquanti[n])$ is {\em ex ante feasible} if there
is a distribution over feasible subsets of agents such that for each
$i$, $\exquanti$ is the (ex ante) probability agent $i$ is in the
subset.  The {\em ex ante optimal mechanism} is the one that maximizes
$\sum_i \revi(\exquanti)$ subject to ex ante feasibility of
$\exquants$.
\end{definition}

\begin{proposition}
\label{prop:exante-reduction}
The ex ante optimal revenues for a general service constrained
environment and its single-dimensional analog are equal and an upper
bound on the (ex post feasible) optimal revenues (which may not be
equal).  If the optimal mechanism is a $\beta$-approximation to the ex
ante optimal revenue in the single-dimensional analog environment,
then the marginal revenue mechanism is a $\beta$-approximation to the
optimal revenue in the original environment.
\end{proposition}

\begin{proof}
The ex ante optimal revenue is defined only in terms of revenue curves
and feasibility for the service constrained environment; therefore, a
general environment and its single-dimensional analog have the same ex
ante optimal revenue.  By \autoref{def:optimal-marginal-revenue} the
(ex post feasible) optimal revenue in the single-dimensional analog is
equal to the (ex post feasible) optimal marginal revenue of the
original environment.  To show the reduction, then, it suffices to
observe that the ex ante optimal revenue is an upper bound on the
optimal revenue in the original environment.  As every ex post
feasible mechanism is ex ante feasible (i.e., the latter is a relaxation
of the former), the observation holds.
\end{proof}


The following single-dimensional agent theorem is an immediate
consequence of results of \citet{Y11}; his results, in fact, gave a
specific (ex post feasible but non-optimal) mechanism that satisfies
the claimed bound.  Of course, then, the (ex post feasible) optimal
mechanism satisfies the bound too.  We obtain our desired result for
general agents as a corollary of this theorem
and \autoref{prop:exante-reduction}.





\begin{theorem}
For matriod environments with single-dimensional linear agents, the
(ex post feasible) optimal mechanism is an $e/(e-1)$ approximation to
the ex ante optimal mechanism; in any $k$-unit environment the bound
improves to $\sqrt{2\pi k}/(\sqrt{2\pi k} - 1)$.
\end{theorem}

\begin{corollary}
In any service constrained matroid environment, the marginal revenue
mechanism is an $e/(e-1)$ approximation to the optimal mechanism; in
any service constrained $k$-unit environment the bound improves to
$\sqrt{2\pi k}/(\sqrt{2\pi k} - 1)$.
\end{corollary}

\paragraph{Downward-closed environments}
\label{sec:downward}

In this section we show that in downward-closed environments and for a
large class of agent preferences, the optimal marginal revenue is a
logarithmic approximation, in the number of agents, to the optimal
revenue.  For example, this class includes quasi-linear preferences.
In contrast to \autoref{sec:agent} where we gave a four approximation
for unit-demand preferences with a product distribution (over
alternatives), the results here apply, for example, to agents with
correlated value distributions over alternatives and to quasi-linear
preferences beyond unit demand.

To show this result we will incorporate the downward closure of the
environment in the single-agent lottery pricing problems.
Specifically, it is without loss of generality for downward-closed
environments to duplicate every non-service outcome and label the
duplicate a service outcome.  This transformation implies that revenue
is monotone in the allocation constraint, i.e., weaker constraints give
no lower revenue.

A summary of the construction in the proof is as follows.  If we
consider allocation constraints with a minimum probability of $2^{-K}$
for allocating to any type, then the allocation constraint can be
partitioned into $K$~pieces such that the highest and lowest
probabilities of allocation in each piece are within a factor of two
of each other.  If the single-agent lottery pricing problems satisfy a
natural scalability property then the revenue of each piece can be
approximated by a $\exquant$ \optpricing\@ scaled appropriately so
that it is dominated by the original allocation constraint.  The
optimal revenue, then, is at most an $O(K)$ multiple of the revenue of
the best such scaled \optpricing.  By downward closure, the optimal
marginal revenue exceeds this revenue and is thus an $O(K)$
approximation.  We obtain a logarithmic approximation by observing
that attention can be restricted to allocation constraints for which
$K \approx \log n$.

Recall that the revenue operator $\Rev{\cdot}$ is concave in its
argument and therefore, for any $\gamma \in [0,1]$,
$\Rev{\gamma \calloc } \geq \gamma \Rev{\calloc} + (1-\gamma)\Rev{0}$,
where $\Rev{0} = \rev(0)$ is the optimal revenue when the agent is
never served.  We assume for simplicity of exposition that $\Rev{0}
= \rev(0) = 0$, i.e., that an agent who is not served generates no
revenue.  The revenue scalability property we need is the opposite of
this inequality, which, if the property holds, must therefore be an
equality.  In fact, revenue scalability can be viewed as a very
permissive relaxation of revenue linearity.

\begin{definition}
An agent is \emph{revenue scalable} if for any $\gamma \in [0,1]$ and
any allocation constraint $\calloc$, the optimal revenue for
$\gamma \calloc$ is equal to $\gamma$ times the optimal revenue for
$\calloc$.
I.e.,
$$
\Rev{\gamma \calloc} = \gamma \Rev{\calloc}.
$$
\end{definition}

For example, as we will show, quasi-linear agents satisfy revenue
scalability, but are not generally revenue linear.  Moreover, if
individual rationality is assumed, which usually implies that the
utility and payment of an agent for non-service outcomes is zero, then
even non-quasi-linear agents are revenue scalable.  These observations
are formalized in the following lemma.

\begin{lemma}
Both (a) quasi-linear agents with no value for non-service outcomes
and (b) agents with no utility and payment for any non-service outcome
are revenue scalable.
\end{lemma}

\begin{proof}
A key property of agents that are quasi-linear or have no utility and
payment for non-service outcomes is that their utility and payment for
any non-service outcome can be arbitrarily scaled upward.  If an
agent's utility and payment for a non-service outcome is zero, then
scaling it upwards is trivial; if an agent is quasi-linear then his
value for a non-service outcome is (minus) his payment and quasi-linearity
requires that scaled payments translate to scaled utility.  Thus, it
suffices to show that agents with scalable utility and payment for
non-service outcomes are revenue scalable.

Consider any allocation constraint $\calloc$ and the optimal lottery
pricing for the scaled constraint $\gamma \calloc$.  Denote by
$\lotteryset$ the set of priced lotteries.  As
$\gamma\calloc(\exquant) \leq \gamma$ for all $\exquant$, the
probability of a service outcome in any of the lotteries of $\lotteryset$
is at most $\gamma$.  The theorem holds if we can define an alternative
set of priced lotteries $\lotteryset'$ that meets the constraint
$\calloc$ where the utility and payment of any type for any lottery is
scaled upward a $1/\gamma \geq 1$ multiple.  This is achieved by
scaling the probability of any service outcome in any lottery upwards
by a $1/\gamma$ multiple (without changing its payment), scaling the remaining probability of
non-service outcomes down (so that the total probability is one), and
scaling the utility and payment for non-service outcomes so that it is
$1/\gamma$ multiple of that for the original lottery (which is possible
by the assumption on scalability for non-service outcomes).  Let
$\gamma\qalloc(\cdot)$ denote the optimal allocation rule for
constraint $\gamma \calloc$; the allocation rule from this
construction is $\qalloc$ and it is feasible for constraint $\calloc$.
\end{proof}

We now show that the marginal revenue approximates the optimal revenue
for revenue-scalable agents in downward-closed service-constrained
environments.

\begin{lemma}
For a revenue-scalable agent, any allocation constraint with minimum
allocation probability $\calloc(1) \geq 2^{-K}$ has revenue
$\Rev{\calloc}$ at most $2K \MargRev{\calloc}$.
\label{downwardlemma}
\end{lemma}

\NOTSTOC{\IFSTOCELSE
{
\section{Proof for downward-closed approximation.}
\label{sec:downward-app}

We now prove the main lemma needed for \autoref{thm:downward} which shows that the marginal revenue mechanism is an $O(\log n)$-approximation in $n$-agent downward-closed environments.
\begin{proof}[of \autoref{downwardlemma}]
}
{
\begin{proof}
}

Let $\revscalar = \Rev{\calloc}$ be the optimal revenue under
allocation constraint~$\calloc$. Let $\qalloc \dominated \calloc$ (where notation
$\qalloc \dominated \hat{\qalloc}$ denotes allocation rules whose
cumulative allocation rules satisfy $\cumalloc(\exquant) \leq
\cumcalloc(\exquant)$ for $\exquant \in [0,1]$; importantly, $\cumalloc(1) = \cumcalloc(1)$ is not required) be the allocation of optimal mechanism subject to $\calloc$. Therefore, $\Rev{\qalloc} = \Rev{\calloc}$. If we prove the claim for $\qalloc$, the proof for $\calloc$ follows because $\Rev{\calloc} = \Rev{\qalloc} \leq 2K \MargRev{\qalloc} \leq 2K \MargRev{\calloc}$, where the last inequality follows by definition of dominance and concavity of the revenue function. Therefore, in the rest of the proof we can assume without loss of generality that the optimal allocation subject to $\calloc$ is $\calloc$ itself.

 Define a sequence of quantiles $0
= \quant_0 \leq \quant_1 \leq \cdots \leq \quant_K=1$ such that
$\calloc(\quant_{j-1})\leq 2\calloc(\quant_{j})$, for $j \in \{1, \ldots,
K\}$.  Define $\revscalar_j$ to be the expected revenue from types that
are mapped to a quantile in $[\quant_{j-1}, \quant_j]$, where the
quantile of a type~$\type$ is the probability that a type drawn at
random has a higher probability of service than that of~$\type$ (as
per \autoref{def:allocation-rule} in \autoref{sec:general}). Therefore, the revenue of the
mechanism is $\revscalar=\sum_{j=1}^K \revscalar_j$. Then there must
exist $\optj$ such that $\revscalar \leq K
\revscalar_{\optj}$. In what follows, we define allocation rules $\qallocalt_j(\cdot)$ for all
$j$, such that $\qallocalt_j\dominated \calloc$ (where notation
$\qalloc \dominated \hat{\qalloc}$ denotes allocation rules whose
cumulative allocation rules satisfy $\cumalloc(\exquant) \leq
\cumcalloc(\exquant)$ for $\exquant \in [0,1]$; importantly, $\cumalloc(1) = \cumcalloc(1)$ is not required), and also $\revscalar_j \leq 2
\MargRev{\qallocalt_j}$. In
particular, for $\optj$ we will have
$\qallocalt_{\optj}\dominated \calloc$, and
$2\MargRev{\qallocalt_{\optj}} \geq \revscalar_{\optj} \geq \revscalar/K$,
which will imply that
\begin{align*}
2\max\nolimits_{\qallocalt \dominated \calloc} \MargRev{\qallocalt} \geq
2\MargRev{\qallocalt_{\optj}} \geq \revscalar/K.
\end{align*}

Define function $\qallocalt_j(\cdot)$ to be $\qallocalt_j(\quant)
= \calloc(\quant_{j+1})$ if $\quant\leq \quant_{j+1}-\quant_j$, and 0
otherwise. Notice that for any $\quant$, we have
$\qallocalt_j(\quant)\leq \calloc(\quant)$, and therefore
$\qallocalt_j\dominated \calloc$, by the definition of dominance in
downward-closed environments.

The main technical component of the proof is to show that, for $\qallocalt_j$ defined above, $\revscalar_j \leq
2\MargRev{\qallocalt_j}$. By construction of $\qallocalt_j$, and recalling that $\calloc(\quant_{j}) \leq 2 \calloc(\quant_{j + 1})$,
\begin{align*}
2\MargRev{\qallocalt_j}
   &= 2\int_{0}^1 \qallocalt_j(\quant)\rev'(\quant)\:\dd\quant \\
   &= 2\calloc(\quant_{j+1}) \rev(\quant_{j+1}-\quant_j) \\
   &\geq {\calloc(\quant_{j}) \rev(\quant_{j+1}-\quant_j)}
\end{align*}

It is therefore sufficient to show that
$\calloc(\quant_{j}) \rev(\quant_{j+1}-\quant_j)\geq \revscalar_j$. Recall
that $\revscalar_j$ is the revenue from types that are mapped to
quantiles in $[\quant_j,\quant_{j+1}]$. Any type in
$[\quant_j,\quant_{j+1}]$ is allocated in ${\calloc}$ with probability
at most $\calloc(\quant_j)$. Now define
$L$ to be the set of lotteries chosen by types in
$[\quant_j,\quant_{j+1}]$, and offer only these lotteries to the
agent.\footnote{Recall that, by the taxation principle, any incentive
compatible mechanism consists of a set of lotteries, from which the
agent chooses the one maximizing her utility.} Notice that types with
quantiles in $[\quant_j,\quant_{j+1}]$ choose the same lottery in
$\lotteryset$ as they did in ${\calloc}$ (whereas other types that
used to choose a lottery either switch to some lottery in~$L$ or no
longer choose one if none in~$L$ give them non-negative utility).  As
a result, the measure of the types that choose some lottery in
$\lotteryset$ is at least $\quant_{j+1}-\quant_j$. Now remove
lotteries from $\lotteryset$, from the one with lowest price, until
the measure of types that choose some lottery is exactly
$\quant_{j+1}-\quant_j$.\footnote{This requires continuity of the type space. We assume continuity for simplicity, but the proof can be easily generalized to handle discrete types.} Call this new set of lotteries
$\lotteryset'$.  Notice that the revenue from $\lotteryset'$ is at
least $\revscalar_j$.  Now recall that all the lotteries in
$\lotteryset$, and therefore $\lotteryset'$, allocate with probability
at most $\calloc(\quant_j)$. Revenue scalability implies that the
revenue of $\lotteryset'$ is at most
$\calloc(\quant_{j}) \rev(\quant_{j+1}-\quant_j)$.

To complete the proof, recall that for downward-closed environments
revenue curves are monotone non-decreasing and so marginal revenues are
non-negative.  Therefore, by the definition of marginal revenue and dominance,
$\MargRev{\calloc} \geq \MargRev{\qallocalt_j}$ for all $j$.
\end{proof}

\STOC{
\begin{proof}[of \autoref{thm:downward}]
Consider an alternative mechanism that runs the optimal mechanism with
probability $1/2$, and otherwise picks an agent at random and outputs
an arbitrary outcome that services that agent, regardless of his type
and without charging him.  Notice that the revenue of the alternative
mechanism is half the revenue of the optimal revenue. Let
$\qalloc_1, \ldots,\qalloc_n$ be the allocation rules for the
alternative mechanism.  Notice also that by construction of the
alternative mechanism, for each $i$ and $\quant\in [0,1]$ we have
$\qalloc_i(\quant) \geq 1/2n$.  Therefore we can
invoke \autoref{downwardlemma} with $K = \log 2n$ to conculde that
that the revenue of the alternative mechanism is at most
\IFSTOCELSE
{
$2\log{n}\sum\nolimits_i \MARGREV_i[\qalloc_i].$
}
{\begin{align*}
&2\log{n}\sum\nolimits_i \MARGREV_i[\qalloc_i]. \qedhere
\end{align*}}
\end{proof}
}
}

\begin{theorem}
\label{thm:downward}
In downward-closed revenue-scalable environments with $n$ agents, the
optimal marginal revenue is a $4\log{n}$ approximation to the optimal
revenue.
\end{theorem}

\NOTSTOC{
\begin{proof}
Consider an alternative mechanism that runs the optimal mechanism with
probability $1/2$, and otherwise picks an agent at random and outputs
an arbitrary outcome that serves that agent, regardless of his type
and without charging him.  This alternative mechanism is obviously incentive
compatible, and its revenue is half of the optimal. Let
$\qalloc_1, \ldots,\qalloc_n$ be the allocation rules for the
alternative mechanism.  Notice also that by construction, for each~$i$ and any $\quant\in [0,1]$ we have
$\qalloc_i(\quant) \geq 1/2n$.  Therefore we can invoke \autoref{downwardlemma} with $K = \log 2n$ to conclude that
the revenue of the alternative mechanism is at most
\IFSTOCELSE
{
$2\log{n}\sum\nolimits_i \MARGREV_i[\qalloc_i].$
}
{\begin{align*}
&2\log{n}\sum\nolimits_i \MARGREV_i[\qalloc_i]. \qedhere
\end{align*}}
\end{proof}
}

\section{Single Dimensional Extension Theorems}
\label{sec:simple}

The marginal revenue approach allows natural generalizations of
techniques developed for single-dimensional linear agent environments.
We will focus here on results for the approximation of optimal
mechanisms by simple mechanisms.  In such a study we are not free to
arbitrarily design the simple mechanism.  Instead, we show that
performance guarantees for simple mechanisms are often obtainable by
relating them to marginal revenue mechanisms.

Consider a variant of the red-or-blue car example from the
introduction.  There are $n$~agents, $k$~cars, and $m$~possible
colors.  The social surplus maximizing mechanism (a.k.a.\@ VCG; see
\citealp{vic-61}; \citealp{cla-71}; and \citealp{gro-73}) selects the
$k$~agents whose values for their favorite color are the highest,
serves these agents, and paints each car as the agent prefers.  We
consider this mechanism simple and practical, and we compare its
revenue against the optimal revenue \citep[cf.][]{HR09}.  Shortly we
will give conditions under which this mechanism is approximately
optimal.

One feature of the VCG mechanism in service constrained environments
is that, ex post, i.e., after agents make reports to the mechanism,
each agent faces a {\em uniform} price over the alternatives.  This
price is equal to the favorite-color value among the other agents.
Thus, in the interim stage each agent faces a distribution over
uniform prices.  We show that the VCG mechanism has near-optimal
revenue in two steps.  First, we show that the VCG revenue is close to
the revenue of the optimal mechanism that only offers agents uniform
prices.  Second, we show that the latter revenue is close to the
optimal revenue by any mechanism.  An important observation is that
the intermediate revenue in between these two steps is the optimal
pseudo marginal revenue with uniform ex ante pseudo pricings
(cf. \autoref{def:pseudo-step-mechanism}).




For the first step of the argument, the gap between the VCG revenue
and the optimal pseudo marginal revenue is governed by the
single-dimensional theory.  Both mechanisms opperate on the type space
given by projection of each unit-demand agent's multi-dimensional type
into the single-dimensional space given by his value for his favorite
alternative.  In particular, the VCG revenue for the unit-demand
agents is equal to its revenue under this single-dimensional
projection, and the optimal pseudo marginal revenue for the
unit-demand agents is equal to the optimal revenue for the projection.
For the second step in our argument, by the theory of agent-based
approximation we developed in \autoref{sec:approx}
(e.g., \autoref{prop:pseudo-approx}), we need only analyze how good the
uniform pricings are, as ex ante pseudo pricings.



We consider a concrete simple case before developing the general
theory.  In the above car-selling example, let $k$ be~$1$, and each
agent's value for each color be i.i.d.\@ (i.i.d.\@ among agents and
i.i.d.\@ across the alternatives).  Since each agent's values for the
alternatives are i.i.d., a uniform price for an agent is also a
uniform virtual price (see \autoref{def:unit-demand-pseudo-step}).
Thus, \autoref{thm:unit-demand} implies that the optimal pseudo
marginal revenue (with uniform ex ante pseudop pricing) is a four
approximation to the optimal revenue.  This constitutes the second
step of our planned argument.  For the first step, the standard
single-dimensional theory.  If the distribution of the i.i.d.\@
unit-demand agent's favorite-alternative value satisfies the
regularity condition of \citet{M81}, then the theorem of \citet{BK96}
implies that, for its single-dimensional analog, the second-price
auction is an $\tfrac{n}{n-1}$ approximation to the single-dimensional
optimal revenue, which is in turn equal to the optimal pseudo marginal
revenue for the unit-demand agents.  Combining the two steps, we have
shown that the VCG mechanism for unit-demand agents is a
$\tfrac{4n}{n-1}$ approximation.  This discussion is formalized and
generalized belowq.


\begin{definition}
\label{def:uniformlinear}
A unit-demand agent is {\em $\beta$ uniformly priceable} if, for any
allocation constraint $\calloc$, a distribution over uniform pricings
gives a $\beta$ approximation to the optimal lottery pricing.  The
{\em uniform ex ante pseudo pricings} are the optimal of these
pricings for ex ante constraints.
\end{definition}

\begin{definition} 
\label{def:favorite-alternative-mechanism}
The {\em favorite-alternative single-dimensional analog} of a
unit-demand service constrained environment is given by projecting the
values of each unit-demand agent to the value of his favorite
alternative.  The {\em favorite-alternative extension} of a
single-dimensional mechanism is a mechanism for unit-demand agents that
simulates the given single-dimensional mechanism on reported
values for favorite alternatives (ignoring the other values).  It serves the winners of the simulation their favorite alternatives at the prices of the simulation.
\end{definition}

\begin{proposition}
\label{lem:reduction}
For any service constrained environment,
unit-demand $\beta$-uniformly-priceable agents, and $\alpha$-approximation
mechanism~$\mech$ for the favorite-alternative single-dimensional analog
environment; the favorite-alternative extension of~$\mech$
is an $\alpha\beta$ approximation for the original environment.
\end{proposition}

\begin{proof}
By construction, the revenue of the favorite-alternative extension
of~$\mech$ in the original environment is equal to the revenue
of~$\mech$ for the favorite-alternative single-dimensional analog
environment.  By assumption of the proposition, this revenue is an
$\alpha$-approximation to the optimal revenue for the
favorite-alternative single-dimensional analog.  This
single-dimensional optimal revenue is equal to the optimal pseudo
marginal revenue (with uniform ex ante pseudo pricings) in the
unit-demand environment.  By the assumption that the unit-demand
agents are $\beta$ uniformly priceable, \autoref{prop:pseudo-approx}
implies that this optimal pseudo marginal revenue is a $\beta$
approximation to the unit-demand optimal revenue.  These bounds
combine to imply that the revenue of the favorite-alternative
extension is an $\alpha \beta$ approximation to the optimal revenue.
\end{proof}

To draw single-dimensional extension theorems as consequences to
\autoref{lem:reduction}, we first claim that unit-demand agents with 
values for each alternative independently drawn from (not necessarily identical) regular distributions are eight uniformly priceable.  After this, we apply 
tools from the single-dimensional theory to provide approximation mechanisms for
the favorite-alternative single-dimensional analog, and draw immediate corollaries.


\paragraph{Uniform Priceability}

As described above, the i.i.d.\@ special case of
\autoref{thm:unit-demand} implies that any unit-demand agent with
i.i.d.\@ values for distinct alternatives is four uniformly priceable.
This result approximately generalizes to non-i.i.d.\@ distributions
that satisfy the regularity condition of \citet{M81} as follows.

\begin{definition}
\label{def:regular} 
A distribution specified by distribution function $\dist$ and density
function $\dens$ is {\em regular} if $\val -
\frac{1-\dist(\val)}{\dens(\val)}$ is monotone non-decreasing in
$\val$.  A single-dimensional linear agent is regular if his value is
drawn from a regular distribution.\footnote{Single-dimensional
  regularity is equivalent (a) to $\rawrev(\exquant) = \rev(\exquant)$ for
  all $\exquant$ (see the proof of \autoref{t:sd=>lin}), and (b) to the
  revenue-optimal allocation rule for $\calloc$ being $\calloc$ itself
  (see \autoref{lem:X=Y}).  These properties of regular distributions
  enable approximation of optimal mechanisms by simple mechanisms in
  single-dimensional environments.} 
\end{definition}

\begin{lemma}
A unit-demand agent with values for alternatives drawn independently from
(not necessarily identical) regular distributions is eight uniformly
priceable.
\label{lem:uniformis8}
\end{lemma}

\begin{proof}[Proof Sketch]
The proof will follow the template given by
\autoref{prop:pseudo-approx} with the following main ingredients.
\begin{itemize}
\item
Twice the optimal revenue of the representative environment (where the
unit-demand agent is replaced by single-dimensional representatives
for each alternative) is a linear upper bound on the optimal revenue
for any allocation constraint (by
\autoref{lem:unit-demand-linear-ub}).
\item Uniform pricing in the representative environment with regular
  distributions gives a four approximation to the optimal
  representative revenue; the argument is as follows.  \citet{HR09}
  show that the second-price auction with a uniform (a.k.a.,
  anonymous) reserve price is a four approximation to the optimal
  revenue.  In fact, this result can be strengthened (a) using a
  prophet-inequality-like proof to give the same bound for uniform
  pricing and (b) to allow an ex ante constraint on the probability
  that any representative is served.  These extensions follow from a
  relatively straightforward modification of
  \autoref{lem:uvpvsoptimal} which we omit.
\item Uniform pricing in the original environment has the same revenue
  as uniform pricing in the representative environment. \qedhere
\end{itemize}
\end{proof}

Below we will make use of the following slight strengthening of the
regularity condition of \autoref{def:regular}.

\begin{definition}
 A unit-demand agent is {\em
  favorite-alternative regular} if the distribution of the agent's
value for favorite alternative is regular; a unit-demand agent is {\em
  individual-alternative regular} if the agent's value for each
alternative is regular; a unit-demand agent is {\em regular} if he is
both favorite- and individual-alternative regular.\footnote{Neither
  favorite-alternative nor individual-alternative regularity imply the
  other.}
\end{definition}

Note that \autoref{lem:uniformis8} requires only individual-alternative
regularity.

\paragraph{Monopoly and Anonymous Reserve Pricing} 
For a single-dimensional single-agent problem, the {\em monopoly
  price} is the price that optimizes revenue.  For single-dimensional,
i.i.d., regular, matroid environments the surplus maximizing mechanism
(a.k.a.\@ VCG) with the monopoly reserve price (for the distribution)
is revenue optimal.  \citet{HR09} approximately extend this result to
non-identical distributions.  They show that with regular
single-dimensional agents, the revenue of the surplus maximizing
mechanism with monopoly reserves is a two approximation to the optimal
revenue.  The following is a corollary of the above development and
their theorems.

\begin{corollary}
For independent unit-demand favorite-alternative-regular
$\beta$-linearly-priceable agents and matroid service-constrained
environments, the surplus maximizing mechanism with monopoly reserves
(for distributions of favorite alternatives) is a $2\beta$
approximation to the optimal revenue. For regular unit-demand agents,
$\beta = 8$.
\end{corollary}

For single-item environments a similar approximation bound holds for
an anonymous reserve price, i.e., one that is the same across the
distinct agents.  \citet{HR09} show that with regular
single-dimensional linear agents in single-item environments, the
revenue of the second-price auction with an appropriate anonymous
reserve is a four approximation to the optimal revenue.  From this
result we obtain the following corollary.

\begin{corollary}
For independent unit-demand favorite-alternative-regular
$\beta$-linearly-priceable agents and single-item service-constrained
environments, the surplus maximizing mechanism with a suitably choosen
anonymous reserve price is a $4\beta$ approximation to the optimal
revenue. For regular unit-demand agents, $\beta = 8$.
\end{corollary}

\paragraph{Market Expansion}

\citet{BK96} show that the revenue of the single-item second-price
auction for $n$ i.i.d.\@ regular agents is at least the revenue of the
optimal auction for $n-1$ agents. An interpretation of this result is
that the revenue loss of running the (surplus-optimal) second-price
auction instead of the revenue-optimal auction can be made up by
recruiting one more agent to the auction. This result generalizes to
matroid environments, see e.g., \citet{DRS09}, where the revenue of
the surplus maximizing mechanism is at least the revenue of the
optimal auction after removing a {\em base} of the matroid.\footnote{A
  base of a matroid is a feasible set with maximum cardinality.}  The
corollary, below, extends the ($k=1$) multi-unit result described
informally in the beginning of this section.

\begin{corollary}
For i.i.d.\@ unit-demand favorite-alternative-regular
$\beta$-linearly-priceable agents and matroid service-constrained
environments, the surplus maximizing mechanism is a $\beta$
approximation to the optimal revenue with any base of the matroid
removed. For regular unit-demand agents, $\beta = 8$.
\end{corollary}

\paragraph{Prior-Independent Mechanisms}

\citet{DRY10} show that, in regular single-dimensional matroid
environments, the surplus maximizing auction where each agent faces a
reserve price randomly drawn from his value distribution is a four
approximation to the optimal auction.  If the agents' values are
identically distributed then the approximation factor improves to two.
Moreover, as long as there are at least two agents with values drawn
from the same distribution, this approximation result can be obtained
by a prior-independent mechanism, i.e., one that is not parameterized
by the prior-distribution.  We summarize the consequences of the
i.i.d.\@ result in general service-constrained matroid environments as
follows.

\begin{corollary}
For i.i.d.\@ unit-demand favorite-alternative-regular
$\beta$-linearly-priceable agents and matroid service-constrained
environments, there is a prior-independent mechanism that is a
$2\beta$ approximation to the optimal revenue.  For regular
unit-demand agents (whose values for alternatives are drawn independently from
not necessarily identical distributions), $\beta = 8$.
\end{corollary}

These results are meant as examples of single-dimensional results with
automatic extensions to unit-demand service constrained environments.
Many other single-dimensional results also can be extended.

\bibliographystyle{apalike}
\bibliography{bibs}

\appendix
\section{Proofs from \autoref{sec:nonmonotone}}
\label{sec:nonmonotone-app}

We now give a procedure for implementing the marginal revenue mechanism
(\autoref{def:mrm}) with general agents.  Recall that in the marginal revenue
mechanism, each agent faces a distribution over \optpricings, where the
distribution is given by marginal revenue maximization over single-dimensional
analog agents having the same revneue curves.  This maximization over
single-dimensional analogs gives rise to an allocation constraint $\callocmr$,
and then the $\exquant$ \optpricing\ occurs with probability
$-\tfrac{\dd}{\dd\exquant} \callocmr(\exquant)$.  We show that this mixture of
\optpricings\ is implementable within the marginal revenue mechanisms family
(\autoref{d:marginal-revenue-family}), and the randomized mapping from type to
quantile (Step~\ref{step:quantile}) in this implementation is efficiently
computable.  

What properties are needed for such a mapping?  First, for each agent, we need
the quantile to be uniformly distributed over $[0, 1]$.  This way, the
distribution over marginal revenues faced by each agent is as if the competing
agents are single-dimensional linear agents with the same revenue curves.  This
guarantees that, if we map an agent's type to a quantile~$\quant$, the
probability that she revenue wins a service in the marginal revenue
mechanism is equal to $\callocmr(\quant)$.  In other words, this property would
designate an allocation probability $\callocmr(\quant)$ to each
quantile~$\quant$, and therefore in order to get the desired allocation rules
for types, we need only to come up with appropriate mappings of types to
quantiles.  Secondly, we
would like the allocation rules obtained by this procedure to match the
allocation rule given by the previously described mixture over \optpricings.  To be specific, recall that each \optpricing\ is derived from optimizing revenue subject to a step function constraint.  The resulting normalized allocation rule may not be a step
function and is in general weaker.  When these \optpricings\ are composed into
the mixture, the resulting allocation rule, which we denote by $\qallocmr$, is
dominated by and not necessarily equal to $\callocmr$.  It is the allocation
rule~$\qallocmr$, and not $\callocmr$, that we would like to produce.

Recall from the discussion of \autoref{def:mrm} that our goal is to implement
the outcome rule $\toutcomemr$.  If we order the types according to
$\ALLOC(\toutcomemr(\cdot))$, we get a natural mapping from
types to quantiles: $\ttoq(\type) \overset{\triangle}{=} \Prx[\type' \sim
\Dist]{\ALLOC(\toutcomemr(\type')) \geq \ALLOC(\toutcomemr(\type))}$.  This mapping
will have the first property\footnote{As before, we break ties appropriately.}, i.e., $\ttoq(\type)$ will be uniformly on $[0, 1]$,
but it does not have the second property.  This is because by definition the
probability of type~$\type$ winning in the marginal revenue mechanism with this
mapping is $\callocmr(\ttoq(\type))$.  Overall, we get the allocation rule
$\callocmr$ and not the weaker $\qallocmr$.  If we could keep the first
property,
then the problem
reduces to the following: given two non-increasing functions $\qallocmr,
\callocmr: [0, 1] \to [0, 1]$, such that $\qallocmr$ is weaker than $\callocmr$
(in the sense that $\cumcallocmr \geq \cumqallocmr$ pointwise), is there a
randomized function $g: [0, 1] \to [0, 1]$, such that $\Ex{\callocmr(g(\quant)}
= \qallocmr(\quant)$ for every $\quant \in [0, 1]$, and $g(\quant)$ is uniform
on $[0, 1]$ when $\quant$ is uniform on $[0, 1]$?  This is a problem addressed
by the theory of \emph{majorization} \citep[see, e.g.\@][]{HLP-29}, and
has a general solution.  In our context, we give a particularly simple interval
resampling procedure that gives this mapping $g$, which is to be composed with
$\ttoq(\cdot)$ for the eventual randomized mapping from types to quantiles.


\IFSTOCELSE { Our construction is related to the \citet*{HLP-29}
  algorithm for computing the doubly stochastic matrix corresponding
  to vector majorization.}
{
%
}

\begin{definition}
\label{def:interval-resampling}
For allocation constraint $\calloc$ and dominated allocation rule
$\qalloc$ satisfying $\cumcalloc(1) = \cumalloc(1)$ on
$\typespacesize$ discrete types, the {\em interval resampling sequence
construction} starts with $\qalloc\super 0 = \calloc$ and calculates
$\qalloc \super {j+1}$ from $\qalloc \super {j}$ while $\qalloc \super
j \neq \qalloc$ as follows.
\begin{enumerate}
\item Find the highest quantile $\quant$ where $\qalloc(\quant) \neq
  \qalloc \super j (\quant)$.
\item Let $\quant' > \quant$ be the quantile at which the line tangent
  to $\cumalloc$ at $\quant$ with slope $\qalloc(\quant)$ crosses
  $\cumalloc\super j$.\footnote{For discrete type, this intersection
    may happen at a quantile $\quant'$ that does not correspond to the
    boundary between two types.  When this happens split the type into
    two types each occurring with the same total probability and with
    the boundary between them at $\quant'$.}
\item The $j$th resampling interval is $[\quant,\quant']$.
\item Let $\qalloc \super {j+1}$ be $\qalloc \super j$ averaged on $[\quant,\quant']$.
\end{enumerate}
\end{definition}

\begin{proposition}
The interval sampling sequence construction gives a sequence of at
most $\typespacesize$ intervals such that the composition of $\calloc$
with the sequence of resamplings applied to $\ttoq(\cdot)$ is equal to $\qalloc$.
\end{proposition}

\begin{proof}
The proof is by induction on $j$ where the $j$th step assumes the
first $j-1$ types, in order of $\ttoq(\cdot)$, satisfy
$\qalloc\super{j-1}(\ttoq(\type)) = \qalloc(\ttoq(\type))$.  Consider
step $j$.  The assumption that $\cumcalloc(1) = \cumalloc(1)$ ensures
that the intersection of the tangent happens at a $\quant' \leq 1$.
The line segment connecting interval $[\quant,\quant']$ of $\cumalloc
\super j$ has slope equal to $\qalloc(\quant)$, by definition.
Therefore, the $j$th step in the construction leaves $\qalloc \super j
(\ttoq(\type)) = \qalloc(\ttoq(\type))$ for the $j$th type.  The
procedure is linear time as both $\calloc$ and $\qalloc$ are, without
loss of generality, piece-wise constant with $\typespacesize$ pieces,
and in each step $\quant$ and $\quant'$ are increasing and at least
one piece from $\calloc$ or $\qalloc$ is processed.
\end{proof}

The final ingredient in the construction of the marginal revenue
mechanism for agents with general types is in converting the
allocation rule back into an outcome rule.  This can be done exactly
as in \citet{AFHHM12}: if an agent with type $\type$ is served by the
allocation rule, sample from service outcomes of $\toutcomemr(\type)$,
otherwise sample from non-service outcomes of $\toutcomemr(\type)$.

\begin{definition}
\label{def:general-marg-rev-mech}
The {\em marginal revenue mechanism} for general agents works as
follows.
\begin{enumerate}
\item Map reported types $\types = (\type_1,\ldots,\type_n)$ of agents
  to quantiles $\quants = (\quant_1,\ldots,\quant_n)$ by, for each
  agent, composing the interval resampling transformation with $\ttoq(\cdot)$.
\item Calculate the marginal revenue of each agent $i$ as $\rev'_i(\quant_i)$.
\item Calculate the set of agents to be served by marginal revenue
  maximization.
\item Calculate outcomes for each agent $i$ as:
\begin{itemize}
\item sample $\outcome_i \sim \toutcomemr_i(\type_i)$ conditioned on $\ALLOC(\outcome_i) = 1$ if $i$ is to be served, or
\item sample $\outcome_i \sim \toutcomemr_i(\type_i)$ conditioned on $\ALLOC(\outcome_i) = 0$ if $i$ is not to be served.
\end{itemize}
\end{enumerate}
\end{definition}

\NOTSTOC{
Note that instead of calculating outcome rules by mixing over step
mechanisms we could, from the allocation constraint $\callocmr$ for an
agent, calculate the optimal mechanism subject to that constraint,
i.e., with outcome rule $\RULE(\callocmr)$ and revenue
$\Rev{\callocmr}$.  The construction above can be invoked with this
outcome rule in place of $\toutcomemr$ without modification; this
change generally improves revenue.}




\NOTSTOC{\section{Proofs from \autoref{sec:budget}}
\label{sec:budget-app}}

The technique for the proof of \autoref{prop:budget} largely comes
from \citet{LR96} and can be viewed as a consequence of that work.  We
remark that the condition of concavity of~$\dist$ (or, equivalently, the monotonicity of
$\dens$), which was not used in the original paper of \citet{LR96},
was in fact needed for their characterization, as correctly pointed
out by \citet{PV08}.

\begin{proof}[\NOTSTOC{Proof }of \autoref{prop:budget}]
For this proof we will only use allocations for types (instead of quantiles),
and to simplify notation we let $\alloc(v)$ be the allocation probability for type~$v$.
Without loss of generality, we assume that the highest valuation in the support
of~$\dist$ is~$1$.
The standard incentive compatibility condition for single-dimensional linear
preferences (monotonicity of the allocation rule and the payment identity) still holds.  
In particular, for $v > v'$, if $\alloc(v) > \alloc(v')$, then
the payment of $v$ is also strictly larger than that of $v'$.
Therefore, if the budget constraint is binding (as we assumed), then
there is a~$\bar v$ such that the allocation probability is a constant
for all types above~$\bar v$, and the payment for all these types is~$B$.
The proposition then states that, in $\exquant$ \optpricing, the allocation
for types smaller than~$\bar v$ is constantly~$0$.

Payment identity states that the payment of type~$v$ is $v \alloc(v)
- \int_{0}^v \alloc(z) \: \dd z$.  We therefore would like to
maximize the objective function
\begin{align}
\label{eq:obj-budget}
\max \int_{0}^{\bar v} \dens(v) \alloc(v) \virt(v) \: \dd v + [1 -
\Dist(\bar v)] \bar v \alloc(\bar v),
\end{align}
where $\virt(v)$ is the standard virtual valuation function $v -
\frac{1 - \Dist(v)}{\dens(v)}$, subject to the constraints:
\begin{align}
\label{eq:constr-budget}
\bar v \alloc(\bar v) - \int_{0}^{\bar v} \alloc(v) \: \dd v & = B, \\
\label{eq:constr-sell-prob}
\int_0^{\bar v} \alloc(v) \: \dd v + [1 - \Dist(\bar v)]
\alloc(\bar v) & = \exquant, \\
\label{eq:constr-nonneg}
\forall v, \quad \alloc(v) & \geq 0, \\
\label{eq:constr-feasible}
\alloc(\bar v) & \leq 1.
\end{align}

We consider the first-order conditions for the above program.  We use
$\delta$ for the Lagrangian variable for the budget
condition~\eqref{eq:constr-budget}; $\lambda$ for the ex ante selling
probability constraint~\eqref{eq:constr-sell-prob}; $\Pi_v$ for
condition~\eqref{eq:constr-nonneg} for each~$v$ ($\Pi_v \leq 0$); $\eta$ for
condition~\eqref{eq:constr-feasible} ($\eta \geq 0$).  The first order condition
gives
\begin{align}
\dens(v) \left[ \virt(v) + \lambda - \frac{\delta}{\dens(v)} \right] +
\Pi_v = 0, \quad \forall v < \bar v; \\
[1 - \Dist(\bar v)] \left[ \bar v + \lambda + \frac{\bar v \delta}{1 - \Dist(\bar
v)}\right] + \Pi_{\bar v} + \eta = 0.
\end{align}

By complementary slackness, for any $v$ such that $\alloc(v) > 0$,
we have $\Pi_v = 0$.  (In particular, $\Pi_{\bar v} = 0$.)  We 
next argue that $\delta$ is negative.  Assume there is a $v < \bar v$ such
that $\alloc(v) > 0$.  Then we have
\begin{align}
\virt(v) + \lambda - \frac{\delta}{\dens(v)} = 0;
\notag \\
\bar v + \lambda + \frac{\bar v \delta}{1 - \Dist(\bar
v)} + \frac{\eta}{1 - \Dist(\bar v)} = 0.
\notag
\end{align}
We can therefore solve for $\delta$:
\begin{align}
\delta = \left[\virt(v) - \bar v - \frac{\eta}{1 - \Dist(\bar v)}
\right]  \Big/ \left[ \frac{\bar v}{1 - \Dist(\bar v)} +
\frac{1}{\dens(v)} \right] < 0.
\end{align}

Now, if for two different $v, v' < \bar v$
such that their allocation probabilities are both strictly positive,
then $\Pi_v = \Pi_{v'} = 0$, and we will have
\begin{align}
\virt(v) - \frac{\delta}{\dens(v)} = \virt(v') -
\frac{\delta}{\dens(v')},
\notag
\end{align}
or
\begin{align}
\label{eq:budget-contradict}
\virt(v) - \virt(v') = \delta \left( \frac{1}{\dens(v)} -
\frac{1}{\dens(v')} \right).
\end{align}

Suppose $v < v'$, then $\dens(v) \geq \dens(v')$ by our assumption.
Since the distribution is regular, we have $\virt(v) \leq \virt(v')$.
Additionally, we know that $\delta < 0$, and so
\eqref{eq:budget-contradict} can hold only if $\dens(v) = \dens(v')$,
but then the equation says $\dens(v) (v - v') + \Dist(v) - \Dist(v') =
0$, which cannot be true since $\Dist(v) < \Dist(v')$.  Therefore
\eqref{eq:budget-contradict} cannot hold under our assumptions.  

So far we have shown that in the optimal solution to the above linear
program, there can be at most one value $v < \bar v$ such that
$\alloc(v) > 0$.  But then lowering $\alloc(v)$ to~$0$ affects
neither the objective function nor the constraints, and so we obtain
a monotone allocation rule. \footnote{As a standard practice, we have
relaxed the monotonicity condition in the formation of the linear
progrma, and only observe that the optimal solution satisfies the
monotonicity condition under the assumptions on the valuation
distribution.}  Therefore the solution to the program gives rise to
an incentive compatible mechanism, which satisfies \autoref{prop:budget}.
\end{proof}

\paragraph{Derivation of \autoref{ex:budget}}

We first derive the $\exquant$ \optpricings\ for $\exquant < 1 - \dist(B) = 1 -
B$.  By \autoref{prop:budget}, a lottery that costs~$B$ is offered, and, when
bought, it sells the item with probability $\pralloc = B + \exquant$.  A type
with value at least $B / (B + \exquant)$ will buy the lottery (and hence wins
with probability $B + \exquant$).  For $\exquant > 1 - B$, the budget does not
bind and the item is sold at a price of $1 - \exquant$; all types with $\val
\geq 1 - \exquant$ wins the item with certainty.  This immediately shows us the
shape of 
$G_{\val}(\exquant)$, the probability of allocation as a
function of~$\exquant$ for a fixed type~$\val$.  From the perspective of a given type $\val \geq B$, $G_{\val}(\exquant)$  jumps starts at
$\exquant = \tfrac B \val - B$, increases linearly with~$\exquant$ and saturates at $\exquant =
1 - B$.  This is depicted in \autoref{fig:largev}.  For $\val < B$, the budget
never binds, and the corresponding $G_{\val}(\exquant)$ is the familiar step
function (\autoref{fig:smallv}).

Calculating the revenue curve is straightforward.  For $\exquant < 1 - B$, $\rev(\exquant)$ is
$B \cdot (1 - \frac B {B + \exquant})$.  Its derivative, i.e., the marginal
revenue, $\frac{B^2}{(B+\exquant)^2}$, is strictly positive.  (Note that, for $B
< 1/2$, this is different from the linear preference case.  There, the marginal
revenue would be negative for $\exquant > 1/2$.)  For $\exquant \geq 1
- B$, $\rev(\exquant)$ is $\exquant(1 - \exquant)$.  Its derivative is positive
for $\exquant < 1/2$ and negative for $\exquant > 1/2$.

By Step~\ref{step:quantile-monotone} of \autoref{def:monotone-marg-rev-mech}, whenever $\val' < B$, its quantile $1 - \val'$ will be
larger with probability~$1$ than the quantile of a type $\val \geq B$, which is distributed between
$\tfrac B \val - B$ and $1 - B$.  Therefore, such smaller $\val'$'s are only
considered when there is no bidder bidding above~$B$; when this happens, since
the budget does not bind, the auction is the optimal one in the linear
preference case, i.e., a second price auction with reserve price~$\tfrac 1 2$.
When there are bidders bidding above~$B$, the way $\quant_i$ is computed in \autoref{ex:budget} is simply sampling by
$G_{\val}(\exquant)$ as stipulated in \autoref{def:monotone-marg-rev-mech}.  When
$i^*$ is the sole bidder bidding above~$B$, she should pay the \optpricing\ at
her critical quantile.  When $B < 1/2$, this critical quantile is $1 - B$, and
she pays $B$; when $B > 1/2$, the critical quantile is $\min \{1/2, 1 - \max_{i
\neq i^*} \vali\}$, and she pays $\max_{i \neq i^*} \{1/2, \vali\}$.  When there
are other bidders bidding above~$B$, the critical quantile for $i^*$ is always
smaller than $1 - B$, and she will pay $B$ in the corresponding \optpricing. A
losing bidder~$i$ bidding above~$B$ faces a critical quantile $\quant_{i^*}$.
We see from \autoref{fig:largev} that if $\quant_{i^*}$ is larger than the
quantile $\tfrac B {\vali} - B$ at which $G_{\vali}(\exquant)$ jump starts, she
will need to make a payment of~$B$.  This happens for $\vali \geq B /
(\quant_{i^*} + B)$.

\begin{figure}[!hb]
\centering
\subfloat[][$\val \geq B$ \label{fig:largev}]{
	\begin{pspicture}(-1, -1)(6,6)
	\psset{yunit=1.6cm,xunit=1.6cm}
	\psaxes[Dy=1, Dx=1, labels=none, ticks=none]{->}(0, 0)(0,0)(3.5, 3.5)
	\psline[linestyle=dashed, linewidth=0.5pt](0, 1)(1, 2)
	\psline[linecolor=red, linewidth=1.2pt](0,0)(1, 0)(1, 2)(2, 3)(3,3)
	\rput(1, -0.2){$ \frac B \val - B$}
	\rput(-0.2, 1){$B$}
	\psline[linestyle=dashed, linewidth=0.5pt](2,0)(2,3)
	\rput(2, -0.2){$1 - B$}
	\psline[linestyle=dashed, linewidth=0.5pt](0, 3)(2,3)
	\rput(-0.2, 3){$1$}
	\psline[linestyle=dashed, linewidth=0.5pt](3, 0)(3,3)
	\rput(3, -0.2){$1$}
	\psline[linestyle=dashed, linewidth=0.5pt](0, 2)(1,2)
	\rput(-0.2, 2){$\val$}
	\rput(3.5, -0.2){$\exquant$}
	\rput(-0.4, 3.5){$G_{\val}(\exquant)$}
	\end{pspicture}
}
\subfloat[][$\val < B$ \label{fig:smallv}]{
	\begin{pspicture}(-1, -1)(6,6)
	\psset{yunit=1.6cm,xunit=1.6cm}
	\psaxes[Dy=1, Dx=1, labels=none, ticks=none]{->}(0, 0)(0,0)(3.5, 3.5)
	\psline[linecolor=red, linewidth=1.2pt](0,0)(2.3, 0)(2.3, 3)(3,3)
	\rput(2.3, -0.2){$1 - \val$}
	\psline[linestyle=dashed, linewidth=0.5pt](0,3)(2.3, 3)
	\rput(-0.2, 3){$1$}
	\psline[linestyle=dashed, linewidth=0.5pt](3, 0)(3,3)
	\rput(3, -0.2){$1$}
	\rput(3.5, -0.2){$\exquant$}
	\rput(-0.4, 3.5){$G_{\val}(\exquant)$}
	\end{pspicture}
}
\caption{
The allocation rules for $\exquant$ \optpricings\ in \autoref{ex:budget}, for a
fixed type~$\val$.
\label{fig:budget-ex}}
\end{figure}

\section{Proofs for unit-demand approximation}
\label{app:unit-demand}

\autoref{thm:unit-demand} is a consequence of the two lemmas below and
\autoref{prop:pseudo-approx}.

\begin{lemma}
\label{lem:unit-demand-linear-ub}
Twice the optimal representative revenue is a linear
  upper bound on the optimal unit-demand revenue.
\end{lemma}

\begin{proof}
Linearity follows simply from the revenue linearity of
single-dimensional linear agents.  We consider the collection of representatives
as a whole (or, say, a single market), and we can ask what is the optimal
revenue from this market given an ex ante selling probability~$\exquant$ or an
allocation constraint~$\calloc$.  Both terms are easy to find.  
Consider the distribution of the maximum
virtual value (or zero if the maximum virtual value is negative) in
the representative environment.  Index this distribution by quantile
as $\maxvv(\exquant)$.  The optimal revenue for any allocation
constraint $\calloc$ is
$\expect[\exquant]{\maxvv(\exquant)\calloc(\exquant)}$ which is linear in
$\calloc$; this follows from the proof that the optimal revenue in
single-dimensional environments is the virtual surplus maximizer.

We now show that, under any allocation constraint, twice the optimal representative revenue upper bounds
the optimal unit-demand revenue.  To do this we will give two auctions
for the representative environment with the allocation constraint
$\calloc$ and show that the sum of these auctions' revenue upper bounds
the optimal unit-demand revenue for the same constraint.  Of course,
the optimal representative revenue in turn upper bounds each of these auctions'
revenue.

A mechanism for the unit-demand problem is simply a lottery pricing,
i.e., it is a set of lotteries~$L$ with a lottery for each type $\type$ taking the
form of $(\payment(\type),\pralloc^1(\type),\ldots,\pralloc^m(\type))$ with $\sum_j
\pralloc^j(\type) \leq 1$.  The semantics of a lottery is that the agent
pays the price~$\payment(\type)$ and then is allocated an alternative~$j$ at random with
probability~$\pralloc^j(\type)$; the semantics of the collection of
lotteries~$L$ is that the agent, upon drawing her type $\type$ from the
distribution, chooses the lottery $(\payment(\type),\pralloc^1(\type),\ldots,\pralloc^m(\type))$ that corresponds to her type.

Given any collection of lotteries~$L$ that satisfies the allocation
constraint $\calloc$ we define two auctions for the representative
environment that have combined revenue at least that of the collection
of lotteries in the unit-demand environment.

The {\em $L$~mimicking auction} considers the profile of values $\vals
= (\val^1,\ldots,\val^m)$ of the representatives and the lottery that would have been selected by the unit-demand agent with
these values.  It serves the representative~$j$ with the highest value
with probability $\pralloc^j(\type)$ and charges her (no matter whether we serve
her or not) $\payment(\type) -
\sum_{j' \neq j} \pralloc^{j'}(\type) \val^{j'} + \threshutil(\vals
\super 2)$ where $\threshutil(\vals \super 2)$ is the expected utility of
the unit-demand agent with valuation profile $\vals \super 2$ which is
$\vals$ with $\val_j$ replaced with $\max_{j'\neq j} \val^{j'}$.
Notice that the utility of the winning representative~$j$ in this
auction is exactly the same as the unit-demand agent less an amount
that is a function only of the values of the other representatives,
$\vals^{-j}$.  As the utility of the unit-demand agent is monotone in
her value for each alternative, the utility each representative has for
winning is positive when she is the highest valued representative and negative
if she is not (and were to misreport and pretend she were).
Therefore, this auction is incentive
compatible, has revenue at least $\payment(\type) - \sum_{j' \neq
j} \pralloc^{j'}(\type) \val^{j'}$ on valuation profile $\vals$ where
$j$ is the highest valued representative, and satisfies allocation
constraint $\calloc$.  For a given valuation profile, call the second
term in the winning agent's payment, $\sum_{j' \neq j} \pralloc^{j'}(\type)
\val^{j'}$, the {\em deficit} of the $L$~mimicking auction.

The motivation for the next auction is that we want to obtain back the
deficit lost by the $L$~mimicking auction.  Notice that the procedure
that charges the highest valued representative the second highest value and
serves with probability $\sum_j \pralloc^j(\type)$ satisfies the
allocation constraint $\calloc$ and more than balances the deficit;
however, it may not be incentive compatible.

The {\em allocation constrained second-price auction} sells to the
highest valued representative at the second highest representative's
value so as to maximize revenue subject to the allocation constraint
$\calloc$ that any representative is served.  Consider the
distribution of the second order statistic of values and let
$\sos(\quant)$ be the value that the $\quant$ quantile of this random
variable takes on.  The optimal revenue obtainable via a second price
auction with allocation constraint $\calloc$ is
$\expect[\quant]{\sos(\quant)\calloc(\quant)}$.  To obtain this
revenue, conditioning on the second highest value being $\val$, with
probability $\calloc(\sos^{-1}(\val))$ we serve the highest valued
representative and charge her $\val$ (only when we serve her).  This auction is incentive
compatible and revenue optimal (in expectation) among all
second-price procedures that meet the allocation constraint.
Therefore, it more than covers the expected deficit of the $L$
mimicking auction.

We have given two incentive compatible auctions for the representative
environment with combined expected revenue exceeding the revenue of
the lottery pricing $L$.  Therefore, twice the optimal representative
revenue is at least the optimal unit-demand revenue.
\end{proof}

\begin{lemma}
The pseudo revenue curve $\pseudorev(\cdot)$ from uniform virtual
pricings for a unit-demand agent 2-approximates the optimal
representative revenue curve (as a function of $\exquant$ for any
$\exquant$-step constraint).
\label{lem:uvpvsoptimal}
\end{lemma}

\begin{proof}
The proof closely follows the standard prophet inequality proofs (see for example \citet{CMS-10}). As in the proof to the previous lemma, we may view the representatives as one
entity and consider its optimal revenue under ex ante constraint on serving.
Denote the optimal representative revenue for the $\exquant$-step
constraint as a function of $\exquant$ by the revenue curve
$\ORR(\exquant)$.  Consider the outcome of the optimal auction for the
representative environment with ex ante service constraint~$\exquant$.
It sets a uniform virtual price (denoted $\vv(\exquant)$) and serves the
agent with the highest virtual value strictly bigger than $\vv(\exquant)$
with probability one.  If the probability that the largest virtual
value is equal to $\vv(\exquant)$ is strictly positive (which might happen if
any virtual value function is constant on an interval, e.g., from
ironing), it probabilistically accepts or rejects the maximum virtual
value when it is equal to $\vv(\exquant)$ so as to serve with the
desired ex ante probability $\exquant$.  The optimal representative
revenue can thus be calculated and bounded as follows.  Let
$(\vvi[1],\ldots,\vvi[m])$ denote the profile of virtual values of the
representatives.
\begin{align*}
\ORR(\exquant) &= \exquant \cdot \vv(\exquant) + \expect{\max\nolimits_i (\vvi-\vv(\exquant))_+}\\
&\le \exquant\cdot\vv(\exquant)+ \sum\nolimits_i\expect{(\vvi-\vv(\exquant))_+}.
\end{align*}
Above, the notation $(\vvi-\vv(\exquant))_+$ is short-hand for
$\max(0,\vvi-\vv(\exquant))$.

Now we show a lower bound on $\pseudorev(\exquant)$ for $\exquant$ that does not require probabilistic
acceptance in the optimal representative auction described above; denote by $Q
\subset [0,1]$ the set of all such
quantiles.  Let $\eventi$ denote the event that $\vvi[j] < \vv(\exquant)$ for all $j\neq i$; our lower bound
on the $\exquant$ \pseudopricing\ revenue will ignore contributions to the virtual surplus
from the case that more than one representative has virtual value at least $\vv(\exquant)$.
\begin{align*}
\pseudorev(\exquant) &\geq \exquant \cdot \vv(\exquant) + \sum\nolimits_i \expect{(\vvi
-\vv(\exquant))_+ \given \eventi} \cdot \prob{\eventi}\\
&\ge \exquant \cdot \vv(\exquant) + (1-\exquant) \cdot \sum\nolimits_i \expect{(\vvi-\vv(\exquant))_+ \given \eventi}\\
&= \exquant \cdot \vv(\exquant) + (1-\exquant) \cdot \sum\nolimits_i \expect{(\vvi-\vv(\exquant))_+}.
\end{align*}
The second inequality followed because $\prob{\eventi}$, the probability of the event that $\vvi[j] <
\vv(\exquant)$ for all $j\neq i$ is not less than the probability that $\vvi[j] < \vv(\exquant)$ for all $j$,
which is $(1-\exquant)$. To extend this lower bound on $\pseudorev(\exquant)$ from $\exquant \in Q$ to all
$\exquant \in [0,1]$, consider inserting a virtual value $\vv' = \vv(\exquant) + \epsilon$ with measure zero
in the distribution. The $\exquant'$ that corresponds to serving this virtual value or higher has revenue
bounded by the formula above but $\vv' \approx \vv(\exquant)$.  Keeping the virtual value constant and
varying $\exquant$ in the formula interpolates a line between the two revenues.
As the \pseudopricings\ are closed under convex combination, this line gives a
lower bound on the $\exquant$ \pseudopricing. Therefore, the bound above on $\pseudorev(\exquant)$ holds for all $\exquant$.

To bound $\ORR(\exquant)$ in terms of $\pseudorev(\exquant)$ we consider
two cases.  When $\exquant \leq 1/2$ these terms can be directly bounded
as the first terms in both bounds are the same and the second terms
are within a factor of two of each other (by assumption $1-\exquant \geq
1/2$).  To show the claim for $\exquant > 1/2$ notice that

\begin{align*}
\ORR(1) &=\expect{\max\nolimits_i (\vvi)_+} \\
&= \hat{v} + \expect{\max\nolimits_i (\vvi)_+ - \hat{v}} \\
&\le \hat{v} + \expect{(\max\nolimits_i (\vvi)_+ - \hat{v})_+} \\
&= \hat{v} + \expect{(\max\nolimits_i \vvi - \hat{v})_+} \\
&\le \hat{v} + \sum\nolimits_i\expect{(\vvi - \hat{v})_+},
\end{align*}

for any $\hat{v}$. As a result, by setting $\hat{v} = \vv(1/2)$ we get

\begin{align*}
\ORR(1) &\le \vv(1/2)+ \sum\nolimits_i\expect{(\vvi - \vv(1/2))_+} \\
&\le 2{\pseudorev(1/2)}.
\end{align*}

From monotonicity of $\ORR$ and $\pseudorev$ we then conclude that for any $\exquant > 1/2$, $\ORR(\exquant) \leq \ORR(1) \leq 2 {\pseudorev(1/2)} \leq  2 {\pseudorev(\exquant)}$.

\end{proof}

\section{Revenue Linearity for Unit Demand Valuations Uniform on Hypercubes}
\label{sec:hypercube-linear}

In this section we show that unit-demand quasi-linear-utility agents whose
values for $m$~alternatives are i.i.d.\@ drawn from $U[0,1]$ are revenue linear.  Recall from \autoref{app:unit-demand} that an incentive compatible mechanism
offers a menu of lotteries to the agent.  Each lottery takes the form of
$(\payment(\type), \pralloc^1(\type), \ldots, \pralloc^m(\type))$, where $\sum_j \pralloc^j(\type) \leq 1$,
with $\payment$ denoting the price of the lottery and $\pralloc_j$ the
probability with which alternative~$j$ is allocated to the agent.  We sometimes write
$\pralloc$ as the vector $(\pralloc^1, \ldots, \pralloc^m)$.
In this section we abuse the notation and use $\util$ to denote a mapping that
maps a type $\type \in \typespace = [0, 1]^m$ to the expected utility of this
type in an incentive compatible mechanism.  We use the following lemma first noted by \citet{Roc85}.

\begin{lemma}
\label{lem:convex-util}
For a quasi-linear-utility agent, a utility function $\util$ corresponds to an incentive compatible mechanism if and only if it is convex. In this case, $\payment(\type) = \nabla \util\cdot
\type - \util(\type)$, and $\pralloc(\type) = \nabla \util(\type)$.
\end{lemma}

In the above lemma $\nabla \util(\type)$ is the gradient of the function
$\util$. Since selling any alternative accounts as a service, by \autoref{lem:convex-util} the allocation of a
type $\type$ is $||\nabla \util(\type)||_1$, the $L_1$~norm of the vector $\nabla
\util(t)$. Let $W$ be the space of convex utility functions $u$, and $c$ the
cost of producing an alternative.  Using \autoref{lem:convex-util}, we can reformulate the problem of revenue maximization under allocation constraint~$\calloc$ as follows:

\begin{align*}
\mathrm{maximize} & \int_{\typespace} [\nabla \util(\type)\cdot \type -
\util(\type)]\dens(\type) \: \dd \type - c \int_{\typespace} \vec{1}\cdot\nabla
\util(\type) \dens(\type) \: \dd \type \\
\mathrm{s.t.} \qquad & \util \in W \\
 & \forall S \subseteq \typespace, \int_{S} ||\nabla \util(\type)||_1 \: \dd \type
\leq \cumcalloc(\dens(S)).
\end{align*}

Recall from \autoref{sec:single-dimensional} the definition of the cumulative
allocation constraint~$\cumcalloc$.  Note also that the second constraint
automatically guarantees the feasibility constraint: for all but a measure zero
set of types, $||\nabla \util(\type)||_1 \leq 1$.  By our assumption,
$\dens(\type)$ is~$1$ everywhere on $[0, 1]^m$.

For any $\type \in \typespace$, define a scaling function $\patht_\type:[0,1]
\rightarrow \typespace$ as $\patht_\type(\alpha) = \alpha\type$.  Then $\patht_\type(0) =
\vec 0$, and $\patht_\type(1) = \type$, for any~$\type$. We now use the gradient theorem and write
\begin{align*}
\forall \type, \util(\type) - \util(0) = \int_0^1 \nabla \util(\patht(\alpha))\cdot
\patht_\type'(\alpha) \: \dd \alpha.
\end{align*}

In a revenue optimal mechanism, $\util(0) = 0$. Also, by definition of~$\patht$,
$\patht'(\alpha) =
\type$. Therefore,
\begin{align*}
\util(\type) = \int_0^1 \nabla \util(\alpha \type)\cdot \type \: \dd \alpha,
\qquad \forall \type \in \typespace.
\end{align*}

Using this, we can rewrite the objective function as 
\begin{align*}
&\int_\typespace \left[ \nabla \util(\type)\cdot (\type-c\vec{1}) - \int_0^1
\nabla \util(\alpha \type)\cdot \type \: \dd \alpha \right] \: \dd \type \newline \\
=& \int_\typespace \nabla \util(\type)\cdot (\type-c\vec{1}) \: \dd \type - \int_0^1
\int_\typespace \nabla \util(\alpha\type)\cdot \type \: \dd \type \: \dd \alpha.
\end{align*}

In the second term, change variables by defining $v = \alpha\type \in [0, 1]^m$.
Notice that $\type = v/\alpha$, and $\dd v^j = \alpha \: \dd\type^j$ for any $1\leq j \leq m$.
Therefore $\dd v = \alpha^m \: \dd \type$.  Define $\typespace_\alpha$ to be the set
of $\type\in \typespace$ such that $\max_{j} \type^j \leq \alpha$. The objective
is now rewritten as
\begin{align*}
&\int_\typespace \nabla \util(\type)\cdot (\type-c\vec{1}) \: \dd \type -
\frac{1}{\alpha^m} \int_0^1 \int_{v\in \typespace_\alpha} \nabla \util(v)\cdot (v/\alpha) \: \dd v \: \dd \alpha \\
=& \int_\typespace \nabla \util(\type)\cdot (\type-c\vec{1}) \: \dd \type -
\int_{v\in \typespace} \nabla \util(v)\cdot v \int_{\alpha=\max_j{v^j}}^1
\frac{1}{\alpha^{m+1}} \: \dd \alpha \: \dd v \\
=& \int_\typespace \nabla \util(\type)\cdot (\type-c\vec{1}) \: \dd \type -
\int_{v\in \typespace} \nabla \util(v)\cdot v \left[ \frac{1}{m(\max_j v^j)^m}
- \frac{1}{m} \right] \: \dd v \\
=& \int_\typespace \nabla \util(\type)\cdot \left[ \type \left(\frac{m+1}{m} -
\frac{1}{m (\max_j \type^j)^m} \right) - c\vec{1}\right] \: \dd \type.
\end{align*}

Now, if we relax the convexity constraint, the optimization problem is expressed
solely in terms of the gradient of~$\util$.  Next we argue that the optimal
solution to this optimization problem takes a particularly simple form.  First
note that the function $\type^j (\tfrac{m+1}{m} - \tfrac{1}{m \type^j})$ is
increasing in~$\type^j$.  Consider any feasible solution $\nabla \util$ to the program
and its alteration $\nabla \tilde \util$ in the following manner: at any
type~$\type$ where $j^*$ is $\argmax_j \type^j$, let $\nabla^{j^*} \tilde \util$
be $\sum_j \nabla^j \util$, and $\nabla^{j} \tilde \util$ be $0$ for all $j \neq
j^*$.  Since this alteration keeps the $L_1$-norm of $\nabla \util$, $\nabla
\util^*$ still satisfies all the constraints (except that we are relaxing the
convexity constraint for now).  But the objective function is pointwise better
for $\nabla \tilde \util$ than for $\nabla \util$.  Therefore, it suffices to
consider solution gradients whose coordinates at each type~$\type$ are all zero except
the one coordinate where the valuation is maximized (ties can be broken
arbitrarily).  But then the problem degenerates, and the optimal utility
function is given by a simple greedy procedure, which grows, at each type, in the
direction of the maximum valued alternative as much as allowed by the allocation
constraint~$\calloc$.  Formally, the optimal utility function is given by


\begin{align*}
    \util^*(\type) &=
        \begin{cases}
            0, & \max_j \type^j \leq \hat{\type}^c \\
            \int_{\alpha=\hat{\type}^c}^{\max_j \type^j} \calloc(1-\alpha^m)\: \dd \alpha, & \max_j \type^j > \hat{\type}^c,
        \end{cases}
\end{align*}

where $\hat{\type}^c$ solves
\begin{align*}
\hat{\type}^c \left(\frac{m+1}{m} - \frac{1}{m (\hat{\type}^c)^m} \right) = c.
\end{align*}

In particular, $\hat{\type}^0 = \sqrt[m]{\frac{1}{m+1}}$. This utility function
$\util$ specified above is convex and linear in~$\calloc$.  By
\autoref{lem:convex-util}, it is easy to see that $\util^*$ being linear implies
that $\Rev{\cdot}$ is also linear (noting that integral is a linear functional).

To summarize, we have shown that the ex ante optimal mechanism for constraint
$\exquant$ is to post a price of $\sqrt[m]{1-\exquant}$ for any service. The
quantile of each type $\type=(\type^1,\ldots,\type^m)$ is $\quant = 1-(\max_i\type^i)^m$ (see \autoref{fig1}).

\begin{figure}
\centering
\begin{tikzpicture}[scale = 5, align=center]
\filldraw[color=blue!20] (0.6,0.6) -- (1,1) -- (1,0) -- (0.6,0) -- (0.6,0.6);
\filldraw[color=red!20] (0.6,0.6) -- (1,1) -- (0,1) -- (0,0.6) -- (0.6,0.6);
\draw[thick] (0,0) -- (1,0) -- (1,1) -- (0,1) -- (0,0);
\filldraw (0.6,0.3) circle(0.1mm) node[left]{$\type$};
\draw[thick] (0.6,0) -- (0.6,0.6) -- (0,0.6);
\draw[thick] (0.6,0.6)-- (1,1);
\end{tikzpicture}
\caption{The quantile of type $\type = (\type^1,\type^2)$ with $\type^1 \geq \type^2$ is $\quant = 1-(\type^1)^2$.}
\label{fig1}
\end{figure}
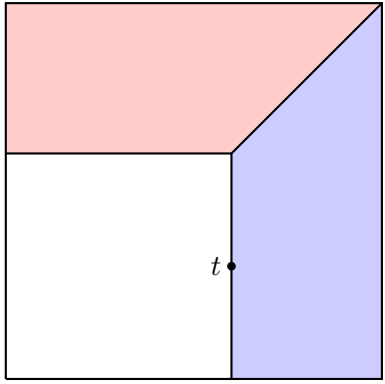

\section{Reverse Auctions}
\label{app:procurement}

Reverse auctions can naturally be modeled as service constrained environments. Different agents, here
sellers, have different costs for providing different services, and the auctioneer has possibly different
values for different services, and wishes to acquire at most one service, and to do so in order to maximize
the value for the service acquired minus the payment for the service. Notice that the goal of maximizing
value minus payment is equivalent to minimizing payment minus value. Such an objective can be modeled as a
forward auction in which the seller has possibly different costs for selling items. In \autoref{sec:hypercube-linear}
we solve the forward auction problem with uniform values and uniform costs, which implies the following
results for the reverse auction problem.

More formally, we can transform a reverse auction problem into a forward auction as follows. Consider a single seller that can provide $m$ services where each service $i$ costs $c_i$. Assume that the cost of each service is drawn uniformly at random from the interval $[0,1]$, and assume that the value of each service for the auctioneer is $1$ (the analysis generalizes to arbitrary distributions and valuations, but the general analysis is not required here). Let $\pralloc(c) = (\pralloc^1(c),\ldots,\pralloc^m(c))$ be the vector of the probabilities of purchasing each service when the cost vector is $c$, and $\payment(c)$ the payment made to the seller by the auctioneer. Now the objective is to maximize
\begin{eqnarray*}
\int_{c\sim U[0,1]^m} \vec{1}\cdot \pralloc(c) - \payment(c) \: \dd c.
\end{eqnarray*}

We can also write the incentive compatibility constraint as
\begin{eqnarray*}
\payment(c) - c \cdot \pralloc(c) \geq \payment(c') - c\cdot \pralloc(c')
\end{eqnarray*}

for all cost vectors $c$ and $c'$. Now define functions $\bar{\pralloc}$ and $\bar{\payment}$ to be
\begin{eqnarray*}
\bar{\pralloc}(c) &=& \pralloc(\vec{1}- c) \\
\bar{\payment}(c) &=& \vec{1} \cdot \pralloc(\vec{1}- c) - \payment(\vec{1}- c). 
\end{eqnarray*}

Using the above notation we can rewrite the objective to be
\begin{eqnarray}
\int_{c \sim U[0,1]^m} \bar{\payment}(c) \: \dd c.
\label{reversetoforwardrev}\end{eqnarray}

Also, 
\begin{eqnarray*}
\payment(c) - c \cdot \pralloc(c) &=& (\vec{1} - c) \cdot \bar{\pralloc}(\vec{1} - c) - \bar{\payment}(\vec{1} - c) .
\end{eqnarray*}

Therefore, the incentive compatibility constraint is equal to
\begin{eqnarray}
c \cdot \bar{\pralloc}(c) - \bar{\payment}(c) \geq c \cdot \bar{\pralloc}(c') -
\bar{\payment}(c'),  \qquad \forall c, c'.
\label{reversetoforwardic}\end{eqnarray}

Now notice that the optimization problem given by~\eqref{reversetoforwardrev} and \eqref{reversetoforwardic} is equal to the standard formulation of a forward auction. We can therefore solve the reverse auction problems by transforming them into forward auction problems, solving the problem using our framework, and then transforming the solution back to the reverse auction setting.

In a reverse auction problem, classical auction theory says that (a) the optimal way to buy an object
(henceforth: a bridge) with value $1$ from a single agent with cost drawn from a
 uniform distribution on $[0,1]$ is to offer a take-it-or-leave-it payment
of $1/2$, (b) the optimal way to buy a bridge with value $1/2$ from one of multiple agents with uniformly
distributed costs is to run a second-price reverse auction with reserve price $1/2$, in which the agent with
the lowest cost (if it is less than $1/2$) constructs the bridge and is payed the minimum of the second lowest
cost and 1/2. The above interpretation of the
 marginal revenue mechanism in i.i.d.\@ settings is one of the most important
result in classical auction theory.  Our theory generalizes this to
 multi-dimensional preferences as follows.  Consider instead buying a
bridge that can be built using technology 1 or technology 2. It says that (a) the optimal way to buy a bridge
with value $1$ from a single
 agent with costs for the different technologies each drawn independently
 and uniformly from $[0,1]$ is to offer a take-it-or-leave-it payment of
$1-\sqrt{1/3}$ for either technology, (b) the optimal way to buy a bridge with value~$1$ from one of
multiple agents each with i.i.d.\@ uniform costs for each technology is to run the second-price reverse
auction with reserve $1-\sqrt{1/3}$ and allow the winning agent to choose her favorite technology to build the
bridge.

\end{document}